\documentclass[12pt,authoryear,review]{elsarticle}
\usepackage[latin9]{inputenc}
\usepackage{geometry}
\usepackage{verbatim}
\usepackage{url}
\usepackage{amsmath}
\usepackage{amsthm}
\usepackage{amssymb}
\usepackage[singlelinecheck=false]{caption}

\usepackage[unicode=true,
 bookmarks=false,
 breaklinks=false,pdfborder={0 0 1}, colorlinks=false]
 {hyperref}
\hypersetup{
 colorlinks=true, linkcolor= blue, citecolor=blue, urlcolor=blue}

\makeatletter

\providecommand{\tabularnewline}{\\}

  \theoremstyle{definition}
  \newtheorem{defn}{\protect\definitionname}[section]
  \theoremstyle{plain}
  \newtheorem{prop}{\protect\propositionname}[section]
  \theoremstyle{plain}
  \newtheorem{thm}{\protect\theoremname}[section]
  \theoremstyle{plain}
  \newtheorem{cor}{\protect\corollaryname}[section]
  \theoremstyle{plain}
  \newtheorem{lem}{\protect\lemmaname}[section]
  \theoremstyle{plain}
  \newtheorem{rem}{\protect\remarkname}[section]
  \theoremstyle{plain}
  
  \theoremstyle{plain}
  \newtheorem*{con*}{\protect\conditionname}
  \theoremstyle{plain}

\usepackage{scalefnt}
\usepackage{amsfonts}
\usepackage[singlespacing]{setspace}
\usepackage[bottom]{footmisc}
\usepackage{indentfirst}
\usepackage{endnotes}
\usepackage{rotating}
\usepackage{tikz}
\usepackage{pgf}
\usepackage{epstopdf}

\newcommand{\nocontentsline}[3]{}
\newcommand{\tocless}[2]{\bgroup\let\addcontentsline=\nocontentsline#1{#2}\egroup}

\newcommand{\com}[1]{}

\providecommand{\corollaryname}{Corollary}
\providecommand{\definitionname}{Definition}
\providecommand{\lemmaname}{Lemma}
\providecommand{\propositionname}{Proposition}

  \providecommand{\corollaryname}{Corollary}
  \providecommand{\definitionname}{Definition}
  \providecommand{\lemmaname}{Lemma}
  \providecommand{\propositionname}{Proposition}
  \providecommand{\remarkname}{Remark}
  \providecommand{\theoremname}{Theorem}
  \providecommand{\conditionname}{Condition}

\begin{document}
\begin{frontmatter}{}

\title{Strategic Decompositions of Normal Form Games: Zero-sum Games and Potential Games\tnoteref{t1}}

\tnotetext[t1]{\today.
\setstretch{0.8}
The research of S.-H. H. was supported by
the Ministry of Education of the Republic of Korea and the National Research Foundation of Korea (NRF-2016S1A3A2924944). The research of L. R.-B. was supported
by the US National Science Foundation (DMS-1109316). We greatly appreciate comments by the advisory editor and two anonymous referees. Especially, we would like to express special thanks to the anonymous referee who gave detailed suggestions on the revision. We also would like
to thank Murali Agastya, Samuel Bowles, Yves Gu\'{e}ron, Nayoung Kim,
Suresh Naidu, Jonathan Newton, David Nielsen and Seung-Yun Oh for their helpful
comments. Special thanks to William Sandholm who read the early version
(\url{http://arxiv.org/abs/1106.3552}) of this paper carefully and
gave us many helpful suggestions.}

\author[A1]{Sung-Ha Hwang\corref{cor1}}

\ead{sungha@kaist.ac.kr}

\author[A2]{Luc Rey-Bellet}

\ead{luc@math.umass.edu }

\cortext[cor1]{Corresponding author. }

\address[A1]{Korea Advanced Institute of Science and Technology (KAIST), Seoul, Korea}

\address[A2]{Department of Mathematics and Statistics, University of Massachusetts
Amherst, MA, U.S.A.}

\begin{abstract}
\setstretch{0.75}
We introduce new classes of games, called \emph{zero-sum equivalent games} and \emph{zero-sum equivalent potential games}, and prove decomposition theorems involving these classes of games. We say that two games are ``strategically equivalent'' if, for every player, the payoff differences between two strategies (holding other players' strategies fixed) are identical. A zero-sum equivalent game is a game that is strategically equivalent to a zero-sum game; a zero-sum equivalent potential game is a potential game that is strategically equivalent to a zero-sum game. We also call a game ``normalized'' if the sum of one player's payoffs, given the other players' strategies, is zero. We present decomposition results involving these games as component games and study the equilibrium properties of these new games. One of our main results shows that any normal form game, whether the strategy set is finite or continuous, can be uniquely decomposed into a zero-sum normalized game, a zero-sum equivalent potential game, and an identical interest normalized game, each with distinctive equilibrium properties. We also show that two-player zero-sum equivalent games with finite strategy sets generically have a unique Nash equilibrium and that two-player zero-sum equivalent potential games with finite strategy sets generically have a strictly dominant Nash equilibrium.
\end{abstract}

\end{frontmatter}

\thispagestyle{empty}

\newpage

\tocless \section{Introduction\setcounter{page}{1}}

When two people start a joint venture, their interests are aligned.
In the division of a pie or unclaimed surpluses, however, someone's gain always comes at the expense of somebody else. So-called identical interest games---games in which all players have the same payoff function---and zero-sum games serve as polar models for studying these social interactions. Two games can be regarded as ``strategically equivalent'' if, for every player, the payoff differences between two strategies (holding other players' strategies fixed) are identical. That is, in two strategically equivalent games, strategic variables such as best responses
of players are the same.\footnote{See Definition \ref{def:st-eq} and \citet{Monderer96}
and \citet{Weibull95}; see also \citet{Morris04} for best response
equivalence.} Potential games---a much-studied class of games in the  literature---are precisely those games that are strategically equivalent to identical interest games. We also introduce a class of games, called ``normalized'' games, in which
the sum of one player's payoffs, given the other players' strategies,
is always zero.

We are interested in zero-sum games and their variants---(i) games that are strategically equivalent to a zero-sum game, accordingly named \emph{zero-sum equivalent games}, (ii) potential games that are strategically equivalent to a zero-sum game, named \emph{zero-sum equivalent potential games} and (iii) \emph{zero-sum normalized games}. Our interest in zero-sum equivalent games is motivated by their definition being analogous to the definition of potential games. Potential games retain all the attractive properties of identical interest games (e.g., the existence of a potential function) because they are strategically equivalent to identical interest games. Thus, zero-sum equivalent games are expected to possess similar desirable properties of zero-sum games as well. It is well-known that two-player zero-sum games with a finite number of strategies have a mini-max solution which is the same as a Nash equilibrium.\footnote{Recently, there have been generalizations of these properties and characterizations for a special class of $n$-player zero-sum games (\citet{Bregman88, Cai15}).}

To examine these classes of zero-sum related games more systematically, we develop decomposition methods of normal form games and obtain several constituent components belonging to these classes. We also study the equilibrium properties of these new games such as the uniqueness/convexity of Nash equilibria and the existence of a dominant equilibrium. Based on these, we provide two specific applications: (i) equilibrium analysis of two-player finite strategy games and (ii) equilibrium analysis of contest games. In the first application, we illustrate  that decomposition can isolate the effect of component games on the Nash equilibrium of the original game (see Figures \ref{fig:dep-eq} and \ref{fig:perturb}) and show that  the total number of Nash equilibria of a given game, depending on some conditions in terms of its decomposition component games, can be maximal or minimal. In the second application,  the uniqueness of Nash equilibria for rent-seeking games is shown by using the special property of zero-sum equivalent games. This highlights that identifying a zero-sum equivalent game via decomposition facilitates equilibrium analysis.

\begin{table}
\begin{tabular}{|c|c|c|}
\hline
4,4$ $ & -1,1 & 1,-1$ $\tabularnewline
\hline
1,-1 & 2,2 & -2,0$ $\tabularnewline
\hline
-1,1 & 0,-2 & 2,2\tabularnewline
\hline
\end{tabular}\medskip

=$\underbrace{%
\begin{tabular}{|c|c|c|}
\hline
2,2$ $ & -1,-1$ $ & -1,-1$ $\tabularnewline
\hline
-1,-1 & 2,2 & -1,-1\tabularnewline
\hline
-1,-1 & -1,-1 & 2,2\tabularnewline
\hline
\end{tabular}}_{=I}$
+$\underbrace{
\begin{tabular}{|c|c|c|}
\hline
0,0 & -1,1 & 1,-1\tabularnewline
\hline
1,-1 & 0,0 & -1,1\tabularnewline
\hline
-1,1 & 1,-1$ $ & 0,0\tabularnewline
\hline
\end{tabular}}_{=Z}$
+$\underbrace{
\begin{tabular}{|c|c|c|}
\hline
1,1 & 1,0 & 1,0\tabularnewline
\hline
0,1 & 0,0 & 0,0\tabularnewline
\hline
0,1 & 0,0 & 0,0\tabularnewline
\hline
\end{tabular}}_{=B}$
+$\underbrace{
\begin{tabular}{|c|c|c|}
\hline
1,1 & 0,1 & 0,1\tabularnewline
\hline
1,0 & 0,0 & 0,0\tabularnewline
\hline
1,0 & 0,0 & 0,0\tabularnewline
\hline
\end{tabular}}_{=E}$

\protect\protect\caption{\textbf{Illustration of a decomposition.} This example
illustrates one of our main results. Here, $I$ is an identical interest game,
$Z$ is a zero-sum game (the Rock-Paper-Scissors game), $B$ is a
game in which the first strategy is the dominant strategy and $E$
is called a ``non-strategic'' game, in which, for every player, the payoff
differences between two strategies (holding other players' strategies
fixed) are identical. Observe that $I$ and $Z$
are ``normalized'' in the sense that the column sums and row
sums of the payoffs are zeros. }
\label{tab:1}
\end{table}

One of our main decomposition results (Theorem \ref{thm:main}) shows that any normal form game, whether the strategy set is finite or continuous, can be decomposed as follows: (i) an  identical interest ``normalized''
component ($I$ in Table \ref{tab:1}), (ii) a  zero-sum ``normalized''
component ($Z$ in Table \ref{tab:1}), (iii) a zero-sum equivalent potential component---component equivalent to both a zero-sum and identical interest game ($B$ in Table \ref{tab:1}) and (iv) a nonstrategic component ($E$ in Table \ref{tab:1}). Most popular zero-sum games, such as Rock-Paper-Scissors
games and Matching Pennies games, belong to the class of
zero-sum normalized games (see also Cyclic games in \citet{Hofbauer00}).

This study makes the following contributions.
First, we develop a more general way of decomposing normal form games than existing methods. Existing decomposition methods of normal form games, such as in \citet{Sandholm10}, \citet{Candogan2011} and \citet{HandR11}, are limited to finite strategy set games, relying on decomposition methods of tensors (or matrices) or graphs. Our new insights lie in viewing the set of all games as an abstract space---a vector space.
In the vector space of games, we decompose a game into a constituent game and its algebraic complement, and our decompositions correspond to  direct sum representations of the vector space (Theorem \ref{thm:main}). In this space, commuting projection mappings are used to single out subspaces and their algebraic complementary subspaces.\footnote{ \citet{Sandholm10} also uses orthogonal projections to obtain decomposition components; however, the orthogonal projections in that paper cannot extract potential games and zero-sum equilibrium.
The decomposition by \citet{Candogan2011} relies on the Helmholtz decomposition of flows on the graph for games with finite strategy sets and uses Moore inverses of matrix operators. They also provide orthogonal projections onto the potential component space and harmonic component space. However, their definition of potential component games is different from the definition by \citet{Monderer96}(see Appendix \ref{sec:existing}). Besides, decomposition results of these two existing studies are limited to finite strategy games} This approach allows us to obtain decomposition results which hold for an arbitrary normal form game, whether finite or continuous (Theorem \ref{thm:main}). In this way, our method shows a unified and transparent mechanism of decompositions of games, which may be modified to decompose other classes of games.
\com{
In the Hilbert space of games, we decompose games into a constituent game and its orthogonal complement (Theorem \ref{thm:main} (iii), (iv)). A representation of a Hilbert space as an orthogonal sum is useful in the sense that we can naturally characterize a class of games by examining their orthogonal complements. For example, the sufficiency
and necessity of the well-known Monderer and Shapley cycle condition for potential games (Theorem 2.8 in \citet{Monderer96}) can be proved by showing that this condition requires that potential games are orthogonal to all zero-sum normalized games \citep{HandR11}. In addition, the orthogonal structure also allows us to extend the vector space decomposition result, Theorem \ref{thm:main} (ii), to the class of $n>2$ player games with finite strategy games. Our Hilbert space decomposition results also hold for continuous payoff function games (Theorem \ref{thm:main} (iv)). The extension of these Hilbert space results to discontinuous payoff function games involves subtle issues related with the identification of payoff functions that are equal almost everywhere, and we discuss issues in more detail in Section \ref{sec:con}.
}

Second, we provide explicit expressions for projections whose ranges and kernels are subspaces of potential games, zero-sum normalized games, identical interest normalized games, zero-sum equivalent games, and zero-sum equivalent potential games. These explicit formulas provide algorithms to extract the component games from a given game as well as yield tests for potential games and zero-sum equivalent games. Then, the extracted components game can be used to understand the equilibrium properties of the original games (see Section \ref{sec:app}). \com{In this way, our decompositions systematically reveal the underlying equilibrium structure of games of interest. near potential games, explicit formula}

Third, to derive explicit expressions for projections, we find useful characterizations of various games. In particular, we provide \emph{new} characterizations (to our knowledge) of potential games (Proposition \ref{prop:pot}, equations \eqref{eq:pot-con2}, \eqref{eq:potent-char}). These new characterizations, along with existing tests and characterizations, may be used to shed light on the structures of potential games---the topic that we leave for future research.

The remainder of this paper is organized as follows. In Section \ref{sec:main-thm}, we present a basic setup and the main decomposition results, along with an example illustrating strategic equivalence.
Section \ref{sec:zero-eq} presents our results on the equilibrium properties of component games. In Section \ref{sec:app}, we provide the two applications and Section \ref{sec:con} concludes the paper. To streamline the presentation, we relegate most of the proofs to the appendix. Also, in Appendix \ref{sec:existing}, we compare our results to existing decomposition results.

\medskip

\tocless \section{Decomposition Theorems \label{sec:main-thm}}

\tocless \subsection{Basic setup: payoff function space \label{subsec:semi-norm}}
Consider a collection of measurable spaces $S_i$ with a $\sigma$-algebra for $i=1, \cdots, n$. Let the product measurable space with the product $\sigma$-algebra be $S:= \prod_{i=1}^{n} S_i$.  For each $i$, $f^{(i)}:S\to\mathbb{R}$ is a real-valued measurable function and  let $f:=(f^{(1)},$ $f^{(2)},\cdots,f^{(n)})$. An  $n$-player game $(N,S,f)$ is specified by a set of players, $N=\{1,2,\cdots,n\}$; a set of strategy profiles, $S$; and a
payoff function, $f$. Thus, given $N$ and $S$, a game is uniquely identified by a vector-valued
payoff function $f$. For each $i$, let $m_i$ be a measure on $S_i$ and $m$ be the product measure, $dm=\prod_{i=1}^n m_i$.  We succinctly write an $n$-fold integration as follows:
\[
	\int f^{(i)} dm = \int\cdots\int f^{(i)}(s_{1},\cdots,s_{n}) dm_{1}(s_{1})\cdots dm_{n}(s_{n}).
\]
We will consider a set of games with integrable payoff functions and thus define a semi-norm of a payoff function,$\left \Vert \cdot \right \Vert$, as follows
\[
    \left \Vert f \right \Vert := \sum_{i=1}^n \int |f^{(i)}| dm
\]
and consider the following vector space of games:
\begin{equation} \label{eq:L}
\mathcal{L}:=\{f:S\to\mathbb{R}^{n}\textnormal{ measurable and }\left\Vert f\right\Vert <\infty\}.
\end{equation}
 \com{
Since we are interested in square-integrable payoff functions, we introduce a semi-norm of a payoff function and a scalar product of payoff functions as follows:
\begin{equation}\label{eq:semi-norm}
  \left\Vert f\right\Vert :=(\sum_{i=1}^{n}\int\left\vert f^{(i)}\right\vert ^{2}dm)^{1/2},\,\,\,\,\,\left\langle f,g\right\rangle :=\sum_{i}\int f^{(i)}g^{(i)}dm.
\end{equation}

and becomes a Hilbert space under identification of $m$-almost everywhere same functions. In Section \ref{subsec:semi-norm}, we view $\mathcal{L}$ as a vector space.}
 \begin{table}
\center
\scalefont{0.95}
\begin{tabular}{c|l}
\hline
Notation &  Name  \tabularnewline
\hline
$\mathcal{I}$ & \textbf{I}dentical interest games \tabularnewline
$\mathcal{Z}$ & \textbf{Z}ero-sum games \tabularnewline
$\mathcal{N}$ & \textbf{N}ormalized games \tabularnewline
$\mathcal{E}$ & non-strategic games (define \textbf{E}quivalence relation) \tabularnewline
\hline
$\mathcal{Z}+\mathcal{E}$ & \textbf{Z}ero-sum \textbf{E}quivalent games \tabularnewline
$\mathcal{I}+\mathcal{E}$ & potential games (\textbf{I}dentical interest \textbf{E}quivalent games) \tabularnewline
\hline
$\mathcal{B} =$ & \textbf (\textbf{B}oth)  zero-sum equivalent (and) potential games \tabularnewline
$(\mathcal{Z}+\mathcal{E})\cap(\mathcal{I}+\mathcal{E})$   \tabularnewline
\hline
$\mathcal{Z\cap N}$ &  \textbf{Z}ero-sum \textbf{N}ormalized games  \tabularnewline
$\mathcal{I \cap N}$ & \textbf{I}dentical interest \textbf{N}ormalized  games  \tabularnewline
\hline
\end{tabular}
\caption{\textbf{Notation}}
\label{tab:notation}
\end{table}

We assume that $S_i$ is either a finite set or a subset of $\mathbb{R}$ and associate a specific measure with each case as follows. If $S_i$ is a  finite set, we suppose that $m_{i}$ is the counting measure with the natural $\sigma$-algebra of all subsets of $S_i$. If $S_{i}$ is a subset of $\mathbb{R}$, we assume that $S_{i}$ is bounded and choose $m_{i}$ to be the Lebesgue measure restriction on the Borel $\sigma$-algebra of $S_i$. We call a game with finite numbers of strategies a \emph{finite game} and a game with continuous strategy spaces a \emph{continuous game}.
We next introduce several classes of games of interest. Following the standard convention, $s_{-i}$ denotes the strategy profile of all players expect player $i$.

\medskip

\begin{defn} \label{def:space}
We define the following subspaces of $\mathcal{L}$: \\
(i) The space of {\em identical interest games}, $\mathcal{I}$, is defined
by

\setstretch{0.7}
\[
\mathcal{I}:=\{ f\in\mathcal{L}:f^{(i)}(s)=f^{(j)}(s) \text{ for all }i,j   \text{  and for all } s \}\footnote{A common interest game is sometimes used to refer to an identical interest game (see \citet{Sandholm10Book}). However, \citet{Aumann89} define a common interest game as a game that has a single payoff pair that strongly Pareto dominates all other payoff pairs. We appreciate an anonymous referee for pointing out this.} \,.
\]
(ii) The space of {\em zero-sum games}, $\mathcal{Z}$, is defined by
\[
\mathcal{Z}:= \{ f\in\mathcal{L}:\sum_{l=1}^{n}f^{(l)}(s)=0 \text{ for all } s \} \,.
\]
(iii) The space of {\em normalized games},  $\mathcal{N}$, is defined by
\begin{align}\label{eq:normalize}
\mathcal{N}:=\{ f\in\mathcal{L}:\int f^{(i)}(t_{i},s_{-i})dm_{i}(t_{i})=0 \text{ for all } s_{-i}, \text{ for all }i \} \,.
\end{align}
(iv) The space of {\em nonstrategic games}, $\mathcal{E}$, is defined by
\begin{align}\label{eq:passive}
\mathcal{E}:= \{ f \in\mathcal{L} : f^{(i)}(s_i,s_{-i}) = f^{(i)}(s'_i,s_{-i}) \text{ for all } s_i, s'_i, s_{-i}, \text{ for all } i  \} \,.
\end{align}
\end{defn}

\noindent Identical interest games in $\mathcal{I}$ and zero-sum games in $\mathcal{Z}$ are familiar games, with cooperative and competitive interactions, respectively.\com{ and tractable\com{straightforward}
equilibrium analysis is  possible in both classes of games via potential functions and mini-max solutions.} A normalized game is a game in which the sum of one player's payoffs, given the other players' strategies,
is always zero. A non-strategic game in $\mathcal{E}$, (also sometimes called a passive game), is a game in which each player's payoff does not depend on his own strategy choice (\citet{Sandholm10, Candogan2011}).
Thus, each player's  strategy choice plays no role in determining her payoff.  Because of this property,  the players'  strategic relations remain unchanged if we add the payoff of a non-strategic game to that of another game.  This leads us to
the definition of  strategic equivalence (and thus notation, $\mathcal{E}$).\footnote{One can study different strategic equivalences. For example, \citet{Monderer96} introduce the
concept of $w-$ potential games in which the payoff changes are proportional for each player. \citet{Morris04}
also study the best response equivalence of games in which players have the same best-responses. We choose
our definition of strategic equivalence since it is most natural with the linear structure of the space of all games.}

\begin{defn} \label{def:st-eq}
\setstretch{0.8}
We say that game $g$ is \emph{strategically equivalent} to game $f$
if
\[
g=f+h\textnormal{ for some }h\in\mathcal{E}
\]
We write this relation as $g\sim f$.
\end{defn}

Note the following simple, but useful, characterization for non-strategic games:  a function does not depend on a variable if and only if the value of the integral of the
function with respect to that variable gives the same value of the function. We write $|S_i|:=m_{i}(S_{i})$.
\begin{lem} \label{lem:non-st}
A game $f$ is a non-strategic game if and only if
    \begin{equation}\label{eq:pass}
        f^{(i)}(s)=\frac{1}{|S_{i}|}\int f^{(i)}(t_{i},s_{-i})dm_{i}(t_{i}) \text{ for all }i, \text{ for all } s \,.
    \end{equation}
\end{lem}
\begin{proof}
Suppose that $f$ satisfies \eqref{eq:pass}. Then, clearly, $f^{(i)}$ does not depend on $s_i$ for all $i$. Now let $f \in \mathcal{E}$. Then there exist $\zeta$ such that $f^{(i)}(s)=\zeta^{(i)}(s_{-i})$ for all $s$, which does not depend on $s_i$, for all $i$. Thus, by integrating, we see that \eqref{eq:pass} holds.
\end{proof}

As mentioned, it is well-known that two-player zero-sum games have desirable properties and thus, in addition to potential games, we also consider a class of zero-sum equivalent games---games that are strategically equivalent to zero-sum games and a class of zero-sum equivalent potential games---potential games that are strategically equivalent to zero-sum games. Given two subspaces $\mathcal{A}$ and $\mathcal{A}'$, the sum of two subspaces is defined to be
\[
    \mathcal{A} + \mathcal{A}':= \{ f+f' :\,\,f \in \mathcal{A},\,\, f' \in \mathcal{A}' \}.
\]

\begin{defn}\label{def:pot-zero}
We have the following definitions: \\
(i) The space of {\em potential games} (identical interest equivalent games) is defined by
\setstretch{0.9}
\begin{align}
\mathcal{I}+\mathcal{E}.
\end{align}
(ii) The space of {\em zero-sum equivalent games} is defined by
\begin{align}
\mathcal{Z}+\mathcal{E}
\end{align}
(iii) The space of games that is strategically equivalent
to both an identical interest game and a zero-sum game, called {\em zero-sum equivalent potential games}, is denoted by $\mathcal{B}$:
\[
    \mathcal{B} = (\mathcal{I} + \mathcal{E}) \cap  (\mathcal{Z} + \mathcal{E})
\]
\end{defn}
The following example illustrates strategic equivalence.
\bigskip

\noindent \textbf{Example (Strategic equivalence: Cournot oligopoly).}

\noindent Consider a quasi-Cournot oligopoly game with a linear demand function for which the payoff function for the $i$-th player, for $i=1,\cdots,n$, is given by
\[
f^{(i)}(s)=(\alpha-\beta\sum_{j=1}^{n}s_{j})s_{i}-c_{i}(s_{i})\,,
\]
where $\alpha, \beta > 0$, $c_i(s_i) \geq 0$ for all $s_i \in [0, \bar s]$ for all $i$ and for some sufficiently large $\bar s$.\footnote{Here, quasi-Cournot games allow the negativity of the price \citep{Monderer96}. Further, we can choose $\bar s$ to ensure that the unique Nash equilibrium lies in the interior $[0, \bar s]$ as follows. Suppose that $c_i(s_i)$ is linear; that is, $c_i(s_i) = c_i s_i$ for all $i$. We assume that $\alpha > n \min_i c_i - (n-1) \max_i c_i$, which ensures that the Nash equilibrium, $s^*$, is positive. If we choose $\bar s$ such that $(n+1) \beta \bar s > \alpha - n \min_i c_i + (n-1) \max_i c_i$, then $s^*_i \in (0, \bar s)$ for all $i$.
\com{Should I discuss the meaning of this condition and existence?} }
It is well-known that this game is a potential game (\citealt{Monderer96}); i.e., it  is strategically equivalent to an identical interest game.  But it is also strategically equivalent to a zero sum game (if $n \ge 3$) as follows.  To show this, when $n=3$, we write the payoff function as
\begin{equation}\label{eq:cournot-1}
  \begin{pmatrix}f^{(1)}\\
f^{(2)}\\
f^{(3)}
\end{pmatrix}=\underbrace{\begin{pmatrix}
(\alpha-\beta s_{1})s_{1}-c_{1}(s_{1})\\
(\alpha-\beta s_{2})s_{2}-c_{2}(s_{2})\\
(\alpha-\beta s_{3})s_{3}-c_{3}(s_{3})
\end{pmatrix}}_{\text{Self-interaction}}-\underbrace{\begin{pmatrix}\beta s_{1}s_{2}\\
\beta s_{1}s_{2}\\
0
\end{pmatrix}}_{\substack{\text{Interactions}\\
\text{between players 1 and 2}
}
}-\underbrace{\begin{pmatrix}\beta s_{1}s_{3}\\
0\\
\beta s_{1}s_{3}
\end{pmatrix}}_{\substack{\text{Interactions}\\
\text{between players 1 and 3}
}
}-\underbrace{\begin{pmatrix}0\\
\beta s_{2}s_{3}\\
\beta s_{2}s_{3}
\end{pmatrix}}_{\substack{\text{Interactions}\\
\text{between players 2 and 3}
}
}
\end{equation}

The self-interaction term  is strategically equivalent to both an identical interest game and a zero-sum game, as the following two payoffs show
\begin{equation} \label{eq:cournot-eq1}
\begin{pmatrix}
&\boldsymbol{(\alpha-\beta s_{1})s_{1}-c_{1}(s_{1})}&+(\alpha-\beta s_{2})s_{2}-c_{2}(s_{2})&+(\alpha-\beta s_{3})s_{3}-c_{3}(s_{3})\\
&(\alpha-\beta s_{1})s_{1}-c_{1}(s_{1})&+\boldsymbol{(\alpha-\beta s_{2})s_{2}-c_{2}(s_{2})}&+(\alpha-\beta s_{3})s_{3}-c_{3}(s_{3})\\
&(\alpha-\beta s_{1})s_{1}-c_{1}(s_{1})&+(\alpha-\beta s_{2})s_{2}-c_{2}(s_{2})&+\boldsymbol{(\alpha-\beta s_{3})s_{3}-c_{3}(s_{3})}
\end{pmatrix}
\end{equation}
and
\begin{equation} \label{eq:cournot-eq2}
\begin{pmatrix}\boldsymbol{(\alpha-\beta s_{1})s_{1}-c_{1}(s_{1})}-[(\alpha-\beta s_{2})s_{2}-c_{2}(s_{2})]\\
\boldsymbol{(\alpha-\beta s_{2})s_{2}-c_{2}(s_{2})}-[(\alpha-\beta s_{3})s_{3}-c_{3}(s_{3})]\\
\boldsymbol{(\alpha-\beta s_{3})s_{3}-c_{3}(s_{3})}-[(\alpha-\beta s_{1})s_{1}-c_{1}(s_{1})]
\end{pmatrix}.
\end{equation}
The payoffs  in \eqref{eq:cournot-eq1} and \eqref{eq:cournot-eq2} are payoffs for an identical interest game and a zero-sum game, respectively.
They are obtained from the self-interaction term by adding payoffs that do not depend on the player's own strategy, and thus are strategically  equivalent (see Definition  \ref{def:st-eq}).

In a similar way, the payoff component describing
the interactions between players $1$ and $2$ is strategically equivalent
to the payoff for an identical interest game and the payoff for a zero-sum
game. For example,
\begin{equation} \label{eq:cournot-eq3}
\begin{pmatrix}
\boldsymbol{\beta s_{1}s_{2}}\\
\boldsymbol{\beta s_{1}s_{2}}\\
0
\end{pmatrix},\,\,
\begin{pmatrix}
\boldsymbol{\beta s_{1}s_{2}}\\
\boldsymbol{\beta s_{1}s_{2}}\\
\beta s_{1}s_{2}
\end{pmatrix},\,\,
\text{ and }
\begin{pmatrix}
\boldsymbol{\beta s_{1}s_{2}}\\
\boldsymbol{\beta s_{1}s_{2}}\\
-2\beta s_{1}s_{2}
\end{pmatrix}
\end{equation}
are all strategically equivalent.  A similar computation holds for the last two terms in equation \eqref{eq:cournot-1} involving the interactions between players $1$ and $3$ as well as between players $2$ and $3$. As a consequence, the quasi-Cournot oligopoly game is strategically equivalent
to both an identical interest game and a zero-sum game.

\tocless \subsection{Decomposition results \label{subsec:vec-dec}}

In this section, we present our main decomposition results. For the convenience of readers, we also present some basic properties of projection operators in a vector space in Appendix \ref{appen:decomp_games}. In the context of game theory, the following two kinds of decompositions receive much attention in the literature: (i)  identical interest  games versus zero-sum games (\citet{Kalai10}) and (ii) normalized games versus non-strategic games (\citet{HandR11, Candogan2011}):
\begin{equation}\label{eq:1st-two-d1}
  \text{(i)}\,\,\, \mathcal{L}=\mathcal{I}\ensuremath{\oplus}\mathcal{Z}, \qquad \text{(ii)}\,\,\, \mathcal{L}=\mathcal{N}\oplus\mathcal{E}
\end{equation}
\com{(see the appendix Corollary in \ref{cor:pre-decomp}.)
Thus, zero-sum games (normalized games) are algebraic complements of identical interest games (non-strategic games) and vice versa.}
where $\oplus$ denotes the direct sum in which every element in $\mathcal{L}$ can be \emph{uniquely} written as the sum of one element in $\mathcal{I}$ (or $\mathcal{N}$) and another element in $\mathcal{Z}$ (or $\mathcal{E}$).

To explain how projections on $\mathcal{L}$ induce decompositions in \eqref{eq:1st-two-d1}, we first introduce  an operator which averages one player's payoffs with equal weights, given all other players' strategies. More precisely, for a scalar valued function $u: S \rightarrow \mathbb{R}$, we introduce operator $T_i$ which acts on scalar valued functions:
\begin{equation}\label{eq:def-t}
  (T_i u)(s) = \frac{1}{|S_i|} \int u(t_i, s_{-i}) dm_i(t_i)
\end{equation}
for each $i$ (see Lemma \ref{lem:non-st}). Then $T_i u$ does not depend on $s_i$.
Note that if we define the following operators on $\mathcal{L}$, which act on vector valued functions,
\begin{equation}\label{eq:symm_opt}
    \mathbf{S} f := ( \frac{1}{n} \sum_{i=1}^n f^{(i)}, \cdots, \frac{1}{n} \sum_{i=1}^n f^{(i)}), \quad \,\,\,\,\mathbf{P} f := (T_1 f^{(1)}, \cdots, T_n f^{(n)})
\end{equation}
then $\mathbf{S}$ and $\mathbf{P}$ are projections on $\mathcal{L}$ (see Lemmas \ref{lem:ti} and \ref{lem:proj-space}). Then, the decompositions in \eqref{eq:1st-two-d1} can be expressed as the ranges and kernels of projections as follows:
\begin{equation}\label{eq:two-d1-oper}
  \mathcal{L} = R(\mathbf{S}) \oplus K(\mathbf{S}),  \qquad \mathcal{L} = K(\mathbf{P}) \oplus R(\mathbf{P})
\end{equation}
where $R$ and $K$ denote the range and kernel of the operators, respectively. That is, the spaces of identical interest games and zero-sum games are the range and kernel of the operator $\mathbf{S}$, while the spaces of non-strategic games and normalized games are the range and kernel of the operator $\mathbf{P}$ (by Lemma \ref{lem:non-st}). This shows how a given projection induces a decomposition of the space of games into the range and kernel of the projection.

We would like to extend decompositions in \eqref{eq:1st-two-d1} (and \eqref{eq:two-d1-oper}) to decompositions involving the subspaces of various games defined in Section \ref{subsec:semi-norm} and Table \ref{tab:notation}. Hence, our first task is to find projections onto these subspaces, and then the decomposition results are induced by these projections. To streamline the presentation, we first state our decomposition results  and then provide explicit expressions for projections onto the subspaces of various games in the subsequent propositions.
\begin{thm} \label{thm:main}We have the following decomposition results:
\begin{align*}
  \hypertarget{D1}{\textbf{D1}:}\quad  &  \mathcal{L} =(\mathcal{I}\cap\mathcal{N})\oplus(\mathcal{Z}+\mathcal{E}) \\
  \hypertarget{D2}{\textbf{D2}:}\quad   & \mathcal{L}=  (\mathcal{I}+\mathcal{E}) \oplus (\mathcal{Z} \cap \mathcal{N})  \\
  \hypertarget{D3}{\textbf{D3}:} \quad & \mathcal{L}=(\mathcal{I}\cap\mathcal{N})\oplus(\mathcal{Z}\cap\mathcal{ N})\oplus \mathcal{B}
\end{align*}
where we recall $\mathcal{B}=(\mathcal{I}+\mathcal{E})\cap(\mathcal{Z}+\mathcal{E})$
\end{thm}
\begin{proof}
  See Appendix \ref{appen:decomp_games}.
\end{proof}

To explain the idea of decompositions in Theorem \ref{thm:main}, we start with a natural way to represent the payoff function into a sum of payoffs, which aggregates strategic interactions among players in the various subset of $N$.\footnote{For example, \citet{Sandholm10} decomposes an $n$-player finite strategy game into $2^n$ component games in which each subset of players is ``active''. \citet{Ui00} also expresses a potential game as a sum of component games in which payoffs depend only on the subsets of players} Note the following partition of identity $I$
\begin{equation}\label{eq:par-id}
  I = \prod_{l=1}^{n} (T_l + (I-T_l))= \sum_{M \subset N} \prod_{l \not \in M} T_l \prod_{k \in M} (I- T_k),
\end{equation}
which is obtained by expanding the product and using the commutative property  of $T_i$'s in \eqref{eq:def-t}.
Using \eqref{eq:par-id}, $u:S \rightarrow \mathbb{R}$ can be written as
\begin{equation} \label{eq:u_M}
    u= \sum_{M \subset N} u_M \text{ where  } u_M = \prod_{l \not \in M}T_l \prod_{k \in M}(I-T_k)u.
\end{equation}
(see also Proposition 2.7 in \citet{Sandholm10}).
Observe that $u_M$ is normalized with respect to $s_k$ for all $k \in M$ and $u_M$ does not depend on $s_l$ for all $l \not \in M$. That is, $u_M$ normalizes the payoff function $u$ with respect to the strategies of players in $M$ and renders the payoff function $u$ independent of the strategies of players outside of $M$ by integrating out those strategies. Also note that for $u:S \rightarrow \mathbb{R}$,
\begin{equation} \label{eq:lab}
    \sum_{M \ni i } u_M= (I-T_i) u, \qquad \sum_{M \not \ni i } u_M= T_i u.
\end{equation}
where $\sum_{M \ni i}$ (or $\sum_{M \not \ni i }$) is the summation which runs over all subsets of $N$ containing $i$ (or all subsets of $N$ not containing $i$, respectively).

First, consider the subspace of identical interest normalized games $\mathcal{I} \cap \mathcal{N}$ in Theorem \ref{thm:main} \textbf{D1}. From \eqref{eq:1st-two-d1} and \eqref{eq:two-d1-oper}, we have $\mathcal{I} \cap \mathcal{N} = R(\mathbf{S}) \cap R(\mathbf{P})$. If $\mathbf{S}$ and $\mathbf{P}$ were commuting, $\mathbf{SP}$ would be a projection and thus $\mathcal{I} \cap \mathcal{N}=R(\mathbf{S}) \cap R(\mathbf{P})=R(\mathbf{SP})$ would hold. However, since $\mathbf{S}$ and $\mathbf{P}$ do not commute,
$\mathbf{SP}$ is not a projection onto $\mathcal{I} \cap \mathcal{N}$. Fortunately, it turns out that the following projection $\mathbf{G}$ can be used to define a projection  onto the subspace of identical interest normalized games:
\begin{equation} \label{eq:orth-proj}
    \mathbf{G} f := (\prod_{l=1}^n (I - T_l) f^{(1)},\cdots, \prod_{l=1}^n (I - T_l) f^{(n)}).
\end{equation}
The projection $\mathbf{G}$ normalizes each player's payoff function with respect to \emph{all} players' strategies. Also, it is easy to check that $\mathbf{G}$ commutes with $\mathbf{S}$, hence $\mathbf{SG}$ becomes a projection on $\mathcal{L}$, and we also show that the range and kernel of $\mathbf{SG}$ are the subspaces of identical interest normalized games and zero sum-equivalent games, summarized in the following proposition. 
%

\begin{prop}[\textbf{Decomposition D1}] \label{prop:i-normalized}
  We have the following results. \\
 (i) $\mathbf{SG}$ is a projection. \\
(ii) ${\displaystyle   \mathcal{I} \cap \mathcal{N}= R(\mathbf{S}) \cap R(\mathbf{P}) = R(\mathbf{SG})}$ and ${\displaystyle \mathcal{Z} + \mathcal{E}= K(\mathbf{S}) + K(\mathbf{P}) = K(\mathbf{SG})}$.
\end{prop}
\begin{proof}
  See Lemma \ref{lem:proj-space}, Proposition \ref{appen-prop:iinorm} and Proposition \ref{prop:zero-equiv}.
\end{proof}


%
\noindent The characterizations in Proposition \ref{prop:i-normalized} induce Theorem \ref{thm:main} \textbf{D1}, decomposing $\mathcal{L}$ as the direct sum of the range and kernel of the projection, $\mathbf{SG}$.

Next, consider the subspace of potential games, $\mathcal{I}+ \mathcal{E}$. Again, note that $    \mathcal{I}+\mathcal{E} = R(\mathbf{S})+R(\mathbf{P})$. Thus if $\mathbf{S}$ and $\mathbf{P}$ were commuting and $\mathbf{SP=0}$ held, $\mathbf{S+P}$ would be the projection onto $\mathcal{I}+\mathcal{E}$ and  $\mathcal{I}+\mathcal{E} = R(\mathbf{S})+R(\mathbf{P}) = R(\mathbf{S}+\mathbf{P})$ would hold (see the condition  in Lemma \ref{lem:vec-proj}).
Instead, we managed to find new characterizations for potential games, which can be used to derive an explicit expression for the projection onto the potential game subspace (equations \eqref{eq:pot-con2}, \eqref{eq:potent-char}). In Proposition \ref{prop:pot}, we establish equivalence between these new characterizations and the existing test for potential games in \citet{HandRTest15} (equation \eqref{eq:text-pot-con}).

\begin{prop}[\textbf{Potential games}] \label{prop:pot}
The following statements are equivalent. \\
 (i) $\,\,$ $f$ is a potential game. \\
 (ii)
 \begin{equation}\label{eq:text-pot-con}
   (I-T_i)(I-T_j) f^{(i)} =(I-T_i)(I-T_j) f^{(j)} \textrm{ for all } i,j \,.
 \end{equation}
 (iii)
 \begin{equation}\label{eq:pot-con2}
  f^{(i)}_M = \frac{1}{|M|} \sum_{j \in M} f^{(j)}_M \textrm{ and for all }  i \in M, \textrm{ for all non-empty  } M \subset N \,.
\end{equation}
(iv) For all $i$,
\begin{equation}\label{eq:potent-char}
 f^{(i)} =\sum_{M \ni i} \frac{1}{M} \sum_{j \in M} f_M^{(j)} + \sum_{M \not \ni i}  f_M^{(i)} =\sum_{M \ni i} \frac{1}{M} \sum_{j \in M}  \prod_{l \not \in M} T_l \prod_{k \in M} (I- T_k) f^{(j)} + T_i f^{(i)}
\end{equation}
\end{prop}
\begin{proof}
  See Proposition \ref{appen-prop:pot-char}.
\end{proof}
\noindent  Implicit in Proposition \ref{prop:pot} (see equation \eqref{eq:potent-char}) is the following operator $\mathbf{V}$ on $\mathcal{L}$,
\begin{equation}\label{eq:v-proj}
    \mathbf{V} f := (\sum_{M \ni 1} \frac{1}{M} \sum_{j \in M} f_M^{(j)} , \cdots, \sum_{M \ni n} \frac{1}{M} \sum_{j \in M} f_M^{(j)})
\end{equation}
and it is straighforward to check that $\mathbf{V}$ is a projection (i.e., $\mathbf{V}^2 = \mathbf{V}$),  that $\mathbf{V}$ and $\mathbf{P}$ commute and that $\mathbf{V}\mathbf{P}=0$ (see Lemma \ref{lem:proj-space}). Then, Proposition \ref{prop:pot} implies that $\mathbf{V} + \mathbf{P}$ is a projection onto the subspace of potential games, and
we also show that $K(\mathbf V+\mathbf P)$ is equal to the subspace of zero-sum normalized games:
\begin{table}
  \centering
  \begin{tabular}{c|c|c}
    \hline
    Identical interest normalized games & $\mathcal{I} \cap \mathcal{N}$ & $R(\mathbf{SG})$ \\
    Zero-sum normalized games & $\mathcal{Z} \cap \mathcal{N}$ & $K(\mathbf{V+P}$) \\
    \hline
    Potential games & $\mathcal{I}+\mathcal{E}$ & $R(\mathbf{V+P})$ \\
    Zero-sum equivalent games & $\mathcal{Z}+\mathcal{E}$ & $K(\mathbf{SG})$ \\
    Zero-sum equivalent potential games & $\mathcal{B}$ & $K(\mathbf{SG})\cap R(\mathbf{V+P})$ \\
    \hline
  \end{tabular}
  \caption{Summary of projection mappings}\label{tab:projs}
\end{table}
\begin{prop}[\textbf{Decomposition D2}] \label{prop:d2}
  We have the following results. \\
(i) $\mathbf{V}+\mathbf{P}$ is a projection. \\
(ii) ${\displaystyle  \mathcal{I} + \mathcal{E}=  R(\mathbf{S}) + R(\mathbf{P}) = R(\mathbf{V+P})}$ and  ${\displaystyle \mathcal{Z} \cap \mathcal{N}= K(\mathbf{S}) \cap K(\mathbf{P}) = K(\mathbf{V+P})}$.
\end{prop}
\begin{proof}
See  Lemma \ref{lem:proj-space}, Proposition \ref{appen-prop:pot-char}, and  Proposition \ref{prop:zero-norm}.
\end{proof}
\noindent Again, from Proposition \ref{prop:d2} we obtain the decomposition in Theorem \ref{thm:main}, \textbf{D2}.


Finally, to obtain decomposition $\mathbf{D3}$, we show that $\mathbf{SG}$ and $\mathbf{I-(V+P)}$ commute and $\mathbf{SG(I-(V+P)) =0}$. This implies  that   the range of $\mathbf{SG} + \mathbf{I-(V+P)}$ is equal to the direct sum of the identical interest normalized game subspace and the zero-sum normalized game. We also show that the kernel of $\mathbf{SG} + \mathbf{I-(V+P)}$ is the subspace of zero-sum equivalent potential games:
\begin{prop}[\textbf{Decomposition D3}] \label{prop:d3}
We have the following results. \\
(i) $\mathbf{SG}$ and $\mathbf{I-(V+P)}$ commute and $\mathbf{SG (I-(V+P)) = 0}$. \\
(ii) ${\displaystyle\,\,\,\,\,\,\,\,\,\, (\mathcal{I} \cap \mathcal{N}) \oplus (\mathcal{Z} \cap \mathcal{N}) = R(\mathbf{SG}) \oplus K(\mathbf{V+P}) = R(\mathbf{SG + I - (V+P)})}$.  \\
 ${\displaystyle \,\,\,\,\quad \mathcal{B}= (\mathcal{I} + \mathcal{E}) \cap (\mathcal{Z} + \mathcal{E})
=K(\mathbf{SG}) \cap R(\mathbf{V+P}) = K(\mathbf{SG + I - (V+P)})}$.
\end{prop}
\begin{proof}
    See Proposition \ref{appen:vec-decomp} and Lemma \ref{lem:proj-space}.
\end{proof}
\noindent Then, decomposition \textbf{D3} again follows from the range and kernel decomposition of projection $\mathbf{SG + I - (V+P)}$. We summarize the ranges and kernels of these operators as follows in Table \ref{tab:projs}.

\begin{table}
  \centering
  \setstretch{1.2}
  \scalefont{1}
  \begin{tabular}{c|c|c}
    \hline
    $\frac{1}{N}\sum_{l=1}^{n}f^{(l)}_N$ & $\sum_{\substack{M \ni i \\ M \neq N}} \frac{1}{M} \sum_{j \in M} f_M^{(j)} + T_i f^{(i)}$ & $(I-T_i)f^{(i)} - \sum_{M \ni i} \frac{1}{M} \sum_{j \in M} f_M^{(j)} $ \\
    \hline
    $=f_{\mathcal{I} \cap \mathcal{N}}$ & $=f_{\mathcal{B}}$ & $=f_{\mathcal{Z} \cap \mathcal{N}}$\\
    \hline
    \multicolumn{2}{c|}{ $= f_{\mathcal{I}+\mathcal{E}}$} &  \\ 
    \hline
     & \multicolumn{2}{|c} {$= f_{\mathcal{Z}+\mathcal{E}}$}   \\ 
    \hline
  \end{tabular}
  \caption{\textbf{Summary of component games}. In the first row, $f^{(i)}$ is decomposed into three components,$f_{\mathcal{I} \cap \mathcal{N}}$, $f_{\mathcal{B}}$ and  $f_{\mathcal{I} \cap \mathcal{N}}$. Then $f_{\mathcal{I}+ \mathcal{E}}$ is obtained by adding the first two components and $f_{\mathcal{Z}+\mathcal{E}}$ is obtained by adding the last two components. }\label{tab:component}
\end{table}

Decompositions in Theorem \ref{thm:main} are of great importance in practice since they provide algorithms to extract the component games from a given game. In other words, the component games can be interpreted as  ``closest'' potential, zero-sum equivalent, zero-sum equivalent potential, identical interest normalized, and zero-sum normalized games to the original games.
Concretely, by applying projections, we obtain each component game as follows:
\[
    f= \underbrace{\mathbf{SG}f}_{\in \mathcal{I} \cap \mathcal{N}}  + \underbrace{(\mathbf{I-(V+P)})f}_{\in  \mathcal{Z} \cap \mathcal{N}} + \underbrace{(\mathbf{V+P - SG}) f}_{\in  \mathcal{B}}.
\]
Table \ref{tab:component} summarizes how to find each component game.

\begin{rem} \normalfont
In the special case of $2$-player games, $n=2$, we have the indentity
\[
    \mathbf V + \mathbf P = \mathbf I - (\mathbf I-\mathbf S) \mathbf G
\]
and all subspaces can be described by using the projections $\mathbf{S, G}$, yielding much simpler characterizations for \textbf{D2} and \textbf{D3}:
   \begin{align*}
       \textbf{D2}' \quad & \mathcal{L} =
     (\mathcal{I}+\mathcal{E})
     \oplus
     (\mathcal{Z} \cap \mathcal{N}) \,\,\,\,\,\,\,\,\,\,\,\,\,=  K(\mathbf{(I-S)G}) \oplus R(\mathbf{(I-S)G})\\
    \textbf{D3}' \quad &  \mathcal{L}  = (\mathcal{I}\cap\mathcal{N})
     \oplus
     (\mathcal{Z}\cap\mathcal{N}) \oplus \mathcal{B} = R(\mathbf{SG}) \oplus R(\mathbf{(I-S)G}) \oplus K(\mathbf{G}).
     \end{align*}
$\square$
\end{rem}

\begin{rem} \normalfont
If we introduce the following scalar product of payoff functions
\begin{equation}\label{eq:semi-norm}
 \left\langle f,g\right\rangle :=\sum_{i}\int f^{(i)}g^{(i)}dm,
\end{equation}
then the decompositions, $\textbf{D1}, \textbf{D2}$ and $\textbf{D3}$, become orthogonal. Orthogonal decompositions are useful in the sense that we can naturally characterize a class of games by examining their orthogonal complements. For example, the sufficiency
and necessity of the well-known Monderer and Shapley cycle condition for potential games (Theorem 2.8 in \citet{Monderer96}) can be proved by showing that this condition requires that potential games are orthogonal to all zero-sum normalized games (see Section \ref{subsec:finite} and \citet{HandR11}).
$\square$
\end{rem}

\medskip

\tocless \section{Equilibrium properties of component games \label{sec:zero-eq}}

\tocless \subsection{Zero-sum equivalent games}
In this section, we discuss the equilibrium properties of component games in Definitions \ref{def:space} and \ref{def:pot-zero}. When we study Nash equilibria of finite games, we will consider both pure and mixed strategies. To this purpose, for finite games we let $\Delta_i \,=\, \{\sigma_i \in\mathbb{R}^{|S_i|}: \sum_{s_i \in S_i} \sigma_{i}(s_i)=1,\,\sigma_{i}(s_i) \geq 0\text{\,\,\ for all } s_i \}$ with $\sigma_i(s_i)$ being the probability that player $i$ uses strategy $s_i$.  We also follow the usual convention of extending the domain of the payoff $f$ from $S$ to $\Delta = \prod_{i=1}^n
\Delta_i$ by defining
\begin{equation}
f^{(i)}(\sigma):=\sum_{s\in S} f^{(i)}(s)  \prod_{k} \sigma_k(s_k). \label{eq:game}
\end{equation}
For continuous strategy games, we consider mainly the set of Nash equilibria in pure strategies (except for Proposition \ref{prop:norm-zero-ci}).\footnote{There are existing results for the sufficient conditions ensuring the existence of a pure strategy Nash equilibrium (\citet{Debreu52}, \citet{Glicksberg52}, \citet{Fan52}, \citet{Dasgupta86}, \citet{Reny03New} and  \citet{Duggan07}). Rather than imposing specific conditions, we simply require that a continuous game possess a pure strategy Nash equilibrium if necessary.} Hereafter, for a continuous game we denote by $s$  a pure strategy profile, while for a finite game we denote by $s=\sigma$ a mixed strategy profile by abuse of notation.

To study the equilibrium properties of zero-sum equivalent games, for a given zero-sum equivalent game $f = w+h$ where $w \in \mathcal{Z}$ and $h \in \mathcal{E}$, we introduce
\begin{equation}\label{eq:phi}
  \Phi_f(s):= \max_{t \in S} \sum_{i=1}^{n} w^{(i)}(t_{i}, s_{-i}) = \sum_{i=1}^{n}\max_{t_i \in S_i} w^{(i)}(t_{i}, s_{-i}).
\end{equation}
The function in \eqref{eq:phi} has been used by various authors to examine the existence of a Nash equilibrium for a game.\footnote{See \citet{Nikaido55}, \citet{Rosen65}, \citet{Bregman88}, \citet{Myerson97}, \citet{Barron08}, \citet{Cai15}} We will show that $\Phi_f$ provides some useful characterizations for the class of zero-sum equivalent games. Note that
\begin{equation} \label{eq:char}
    \Phi_f(s^*) = \min_s \Phi_f(s)=0 \iff  s^* \text{  is a Nash equilibrium for }f=w+h.
\end{equation}

Using \eqref{eq:char}, we will study the conditions under which a zero-sum equivalent game admits a unique Nash equilibrium or a convex set of Nash equilibria. To do this, we will show that the (strict) convexity of $w^{(i)}(s_i, s_{-i})$ in $s_{-i}$ for all $i$ implies the (strict) convexity of $\Phi_f(s)$ in \eqref{eq:phi}. This follows from the simple fact that the value function of a maximization problem (in \eqref{eq:phi}) is convex in a parameter if the objective function itself is convex in that parameter. The same principle yields the convexity of the profit function in the basic microeconomics context since the objective profit function is convex in prices. Then, since the set of optimizers of a (strictly) convex function is convex (singleton), the relationship in \eqref{eq:char} shows that a sufficient condition for the convexity or uniqueness of Nash equilibria is the convexity or strict convexity of $w^{(i)}(s_i, s_{-i})$ in $s_{-i}$ for all $i$.

\begin{prop}[Nash equilibria for zero-sum equivalent games]\label{prop:zero-convex}
Suppose that $f$ is a zero-sum equivalent game, where $f=w+h$; $w$ is a zero-sum game and $h$ is a non-strategic game. Suppose that $f$ has a Nash equilibrium. \\
(i) If $w^{(i)}(s_{i},s_{-i})$ is convex in $s_{-i}$
for all $s_{i}$ for all $i$, the set of Nash equilibria is convex.\\
(ii) If $w^{(i)}(s_{i},s_{-i})$ is strictly convex
in $s_{-i}$ for all $s_{i}$ for all $i$, there exists a unique Nash equilibrium
for $f$.\end{prop}
\begin{proof}
See Appendix \ref{proof:zero-convex}.
\end{proof}

\begin{table}
\centering\scalefont{0.8}

\begin{tabular}{c|c|c}
\hline
 & Properties & Examples\tabularnewline
\hline
Zero-sum equivalent & Convexity/ uniqueness of NE  & contest games \tabularnewline
 games & under some conditions & quasi-Cournot games \tabularnewline
\hline
Zero-sum equivalent & Two-player games: dominant strategy NE & Prisoner's Dilemma \tabularnewline

potential games &  & quasi-Cournot games\tabularnewline

\hline
Zero-sum normalized  & Unique uniform mixed  & Rock-Paper-Scissors games \tabularnewline
games                 & strategy NE &
Matching Pennies games \tabularnewline
\hline

Identical interest  & Uniform mixed strategy NE &
Coordination games
\tabularnewline
normalized  games & &  \tabularnewline

\hline
\end{tabular}
\caption{\textbf{Summary of equilibrium characterizations for game.} In the table, (C) and (F) mean continuous strategy games and finite games, respectively.}
\label{tab:eq-char}
\end{table}

We next further explore the consequences of Proposition \ref{prop:zero-convex} for two-player finite games.
Though a class of two-player finite games, often called bi-matrix games, is one of the simplest classes, in general, it is acknowledged that even bi-matrix games are hard to solve \citep{Savani06}. We also focus on a class of non-degenerate games. There are several notions of non-degeneracy for finite games, depending on contexts and problems---such as equilibrium characterizations and
classifications  of dynamics.\footnote{\citet{WuandJiang52} define an essential game---a game whose Nash equilibria  all change only slightly against a smaller perturbation to the game and show that almost all finite games are essential; i.e., the set of all essential games is an open and dense subset of the space of
games.  \citet{Wilson71} introduces a non-degeneracy assumption regarding payoff matrices (more
precisely tensors) and shows that almost all games have an odd (hence finite) number of Nash equilibria. In the
context of evolutionary game theory, \citet{Zeeman80} also defines a stable game whose dynamic remains
structurally unchanged against a small perturbation. } Since we wish to study the equilibrium properties of two-player finite games, we adopt the following non-degeneracy assumption introduced by \citet{Quint&Shubik97} (see also \citet{Shapley74}). Let $|supp(\sigma_i)|$ be the size of the support of a mixed strategy $\sigma_i=s_i$ for finite games.

\begin{con*}[\hypertarget{con-N}{\textbf{N}}] Suppose that $f$ is a two-player finite game. If $|supp(\sigma_1)| =k$, then there are no more than $k$ pure strategy best responses for player 2 against $\sigma_1$. Similarly, if $|supp(\sigma_2)|=k$, there are no more than $k$ pure strategy best responses for player 1 against $\sigma_2$.
\end{con*}
\noindent Then Lemma 2.2 in \citet{Quint&Shubik97} shows that a two-player finite game has a finite number of Nash equilibria under Condition \hyperlink{con-N}{\textbf{(N)}}. A straightforward consequence  of Proposition \ref{prop:zero-convex} is that, generically,  two-player finite zero-sum equivalent games have a \emph{unique} Nash equilibrium.

\begin{cor}[two-player finite zero-sum equivalent games]
\label{cor:2p-zero-finite}Suppose that $f$ is a two-player finite zero-sum
equivalent game. Then the set of Nash equilibria for $f$
is convex. If $f$ satisfies Condition \hyperlink{con-N}{\textbf{(N)}},
the Nash equilibrium is unique. \end{cor}
\begin{proof}
See Appendix \ref{appen:other-proofs}.
\end{proof}

\tocless \subsection{Zero-sum equivalent potential games and normalized games} \label{subsec:con-zero-equi}

We denote by $\zeta_{l} : S \rightarrow \mathbb R$  a function that does not depend on $s_{l}$,  that is, $\zeta_l :=T_l g$ for some $g$ (recall $T_l$ is defined in \eqref{eq:def-t}). In Section \ref{sec:main-thm}, we show that the quasi-Cournot model is a potential game which is also  strategically equivalent to a zero-sum game (see equations \eqref{eq:cournot-eq1}, \eqref{eq:cournot-eq2} and \eqref{eq:cournot-eq3}). Note that the payoff function in \eqref{eq:cournot-1} can be written as
\footnote{Indeed, we can choose
$\zeta_1 =  -\beta s_2 s_3 + (\alpha - \beta s_3) s_3 - c_3(s_3)$, $\zeta_2=-\beta s_1 s_3 + (\alpha - \beta s_1) s_1 - c_1(s_1)$, $\zeta_3 = -\beta s_1 s_2 + (\alpha - \beta s_2) s_2 - c_2(s_2)$.

}
\begin{equation}\label{eq:rep}
  (f^{(1)}, f^{(2)}, f^{(3)} ) \sim (\zeta_2+ \zeta_3, \zeta_1 + \zeta_3, \zeta_1 + \zeta_2 )
\end{equation}
From equation \eqref{eq:rep}, we find an identical interest game which is strategically equivalent to $f$ as follows:
\[
    (f^{(1)}, f^{(2)}, f^{(3)} ) \sim (\zeta_1+\zeta_2+ \zeta_3, \zeta_1 + \zeta_2 +\zeta_3, \zeta_1 + \zeta_2+\zeta_3 )
\]
since $\zeta_l$ does not depend on $s_l$. Similarly, we can find a zero-sum game which is strategically equivalent to $f$:
\begin{align*}
  (f^{(1)}, f^{(2)}, f^{(3)} ) & \sim (\zeta_2 + \zeta_3 - 2 \zeta_1, \zeta_1 + \zeta_3 - 2 \zeta_2, \zeta_1 + \zeta_2 - 2 \zeta_3 ) \\
   & =(-\zeta_1 + \zeta_2, \zeta_1-\zeta_2, 0) + (-\zeta_1 + \zeta_3, 0, \zeta_1 - \zeta_3)+(0, -\zeta_2 + \zeta_3, \zeta_2 - \zeta_3).
\end{align*}
The following statement makes these observations more general and precise.

\begin{prop}[\textbf{$n$-player zero-sum potential equivalent games}]
\label{prop:n-p-pop} An $n$-player game with is a zero-sum equivalent potential game if and only if
\begin{align}
(f^{(1)},f^{(2)},\cdots,f^{(n)}) & \sim\sum_{l=1}^{n}(\zeta_{l},\zeta_{l},\cdots,\zeta_{l}) \label{eq:pot-zero1} \\
 & \sim\sum_{i<j}(0,\cdots,0,\underbrace{-\zeta_{i}+\zeta_{j}}_{i-\text{th}},0,\cdots0,\underbrace{\zeta_{i}-\zeta_{j}}_{j-\text{th}},0,\cdots,0), \label{eq:pot-zero2}
\end{align}
where $\zeta_{l}(\cdot)$ does not depend on $s_{l}$.
\end{prop}
\begin{proof}
  See Appendix \ref{proof:prop:n-p-pop}.
\end{proof}
\noindent A similar expression to \eqref{eq:pot-zero1} for potential functions  is in \citet{Ui00} (see the potential function in Theorem 3 in the cited paper; Appendix D).
The immediate consequences of Proposition \ref{prop:n-p-pop} for two-player games are as follows.

\begin{cor}[\textbf{Two-player zero-sum equivalent potential games}]
\label{prop:2p-both}  We have the following results: \\
(i) Consider a two-player zero-sum equivalent potential game with
\[
    f=(f^{(1)},f^{(2)}) \sim \sum_{l=1}^2 (\zeta_l, \zeta_l).
\] If $ (s^*_1, s^*_2) \in (\arg \max_{s_1}\zeta_2(s_1), \arg \max_{s_2}\zeta_1(s_2))$ exists, then $(s^*_1, s^*_2)$ is a Nash equilibrium.  \\
(ii)  Suppose that a two-player finite zero-sum equivalent potential game satisfies Condition \hyperlink{con-N}{\textbf{(N)}}. Then the game has a strictly dominant
strategy Nash equilibrium.
\end{cor}
\begin{proof}
    See Appendix \ref{proof:prop:n-p-pop}.
\end{proof}

\noindent Intuitively, when two players have both identical and conflicting interests, the strategic interdependence effects completely offset each other as in the Prisoner's Dilemma game.
Interestingly, \citet{Alger13}, in their study on preference evolution,
identify a set of games in which ``the right thing to do'', is simply to choose a strategy that maximizes one's own payoff (p.2281 in \citet{Alger13}), and games in this set are strategically equivalent to those games in Corollary  \ref{prop:2p-both} (i). In this way, our decompositions
can also be used to provide helpful characterizations for a class of interesting
games in applications.

Finally, we show that every zero-sum normalized game and identical interest normalized game possess a uniform mixed strategy Nash equilibrium. When an identical interest game is normalized, this game is normalized with respect to all players' strategies and player $i$'s interim payoff becomes zero against all other players' uniform mixed strategies. Similarly, when a zero-sum game is normalized, one player's payoff can be expressed as the sum of all the other player's payoffs, each normalized with respect to her own strategy, and this again causes player $i$'s interim payoff to be zero.
\begin{prop}[\textbf{Zero-sum normalized games and identical interest normalized games}]
\label{prop:norm-zero-ci}Suppose that a game is a zero-sum normalized game or an identical interest normalized game. Then the uniform mixed strategy profile is always a Nash equilibrium.\end{prop}
\begin{proof}
See Appendix \ref{proof:norm-zero-ci}.
\end{proof}

\tocless \section{Applications  \label{sec:app}}

\tocless \subsection{Two-player finite strategy games \label{subsec:finite}}
In this section, we present applications of the decomposition results in Section  \ref{sec:main-thm} and equilibrium characterizations in Section \ref{sec:zero-eq}. We first show that a two-player finite game can be uniquely decomposed into component games with distinctive equilibrium properties. Then we show that the total number of Nash equilibria for a given game, depending on some conditions in terms of its decomposition component games, can be maximal or minimal.\footnote{For the maximum number of Nash equilibria of finite games, see \citet{Quint&Shubik97, Quint&Shubik02}, \citet{Savani06}, \citet{MacLennan&Park99}.}

 The following statement shows that a given game can be decomposed into three components: the first one with pure strategy Nash equilibria, the second one with a unique uniform mixed Nash equilibrium and the third one with a dominant Nash equilibrium (see Figure \ref{fig:dep-eq} for an illustration of Theorem \ref{thm:equi-decomp}).

\begin{thm}[Two-player finite strategy games; Nash equilibria]\label{thm:equi-decomp}
\begin{figure}[t]
    \centering
    \includegraphics[scale=0.52]{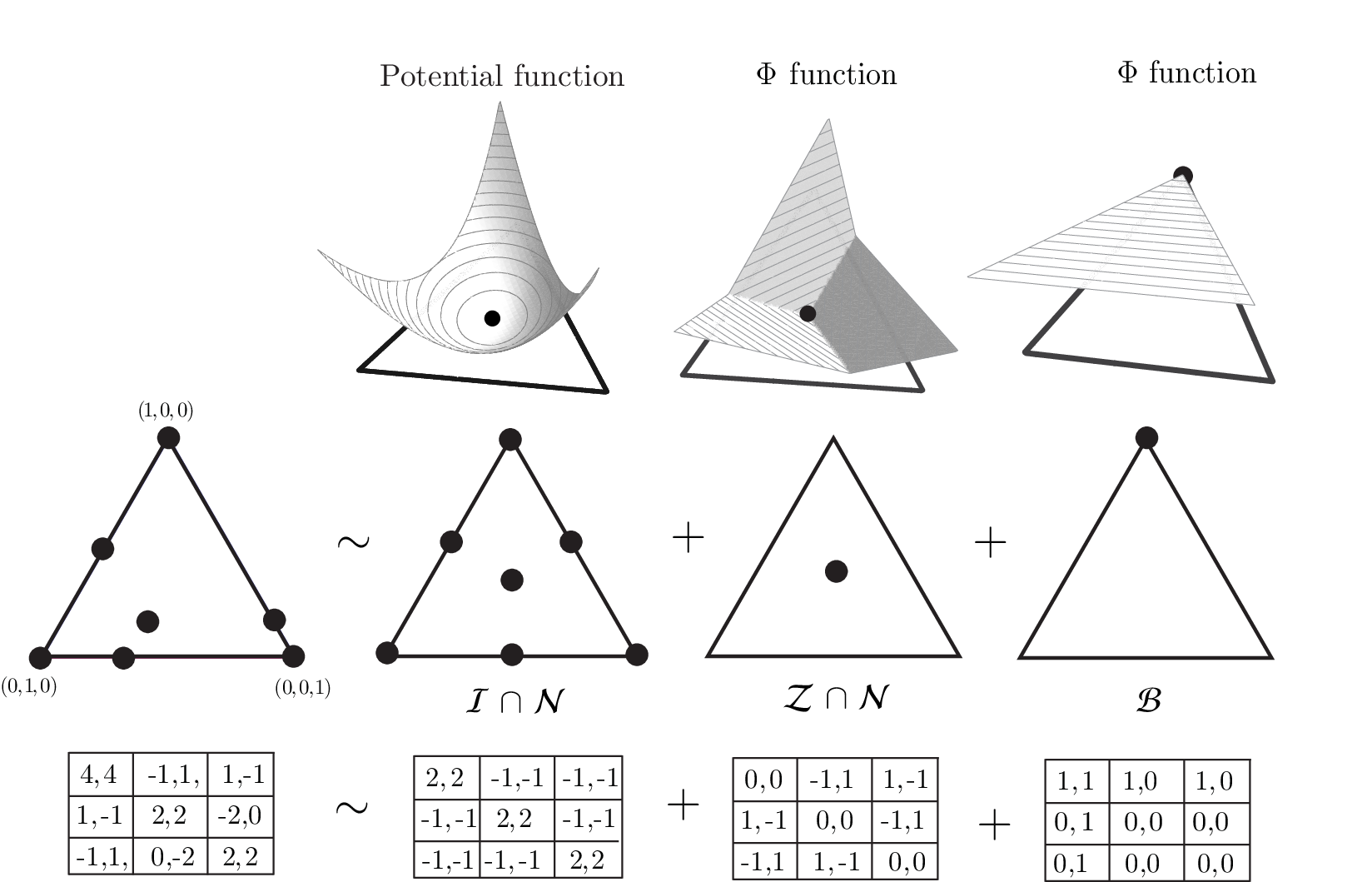}
    \caption{\textbf{Decomposition of a game into components with distinctive Nash equilibria.} In the bottom line we show the two-player game in Table \ref{tab:1}. Since all these games are symmetric, we find all symmetric Nash equilibria for the original games (three pure strategy Nash equilibria, $(1,0,0),(0,1,0), (0,0,1)$, three mixed strategy Nash equilibria involving two strategies $(1/2, 1/2, 0)$, $(1/6, 0, 5/6)$, $(0, 2/3, 1/3)$; and a completely mixed strategy Nash equilibrium $(1/6, 1/2, 1/3)$) and also find all symmetric Nash equilibria for all the other component games. We show these Nash equilibria using the simplex  in the middle line. The top line shows the potential function and the function $\Phi$. The potential function for the zero-sum equivalent potential game in $\mathcal{B}$ is given by $p_1$, while function $\Phi$ is given by $1-p_1$, where $p=(p_1,p_2, p_3) \in \Delta$.}
    \label{fig:dep-eq}
\end{figure}
Suppose that $f$ is a two-player finite strategy game. Then, $f$
can be uniquely decomposed into three components:
\[
f= f_{\mathcal{I}\cap \mathcal{N}}+f_{\mathcal{Z} \cap \mathcal{N}}+f_{\mathcal{B}}
\]
where $f_{\mathcal{I}\cap \mathcal{N}}$ is an identical interest normalized game, $f_{\mathcal{Z} \cap \mathcal{N}}$ is a zero-sum normalized game and $f_{\mathcal{B}}$ is a zero-sum equivalent potential game.
Suppose that all three component games satisfy Condition \hyperlink{con-N}{\textbf{(N)}}. Then, $f_{\mathcal{I}\cap \mathcal{N}}$ has a finite number of Nash equilibria with a uniform mixed strategy, $f_{\mathcal{Z} \cap \mathcal{N}}$ has a unique uniform mixed strategy Nash equilibrium and $f_{\mathcal{B}}$ a the strictly dominant strategy Nash equilibrium.\end{thm}
\begin{proof}
This follows from the decomposition theorem, Theorem \ref{thm:main}, Corollary \ref{cor:2p-zero-finite}, Corollary \ref{prop:2p-both} and Proposition \ref{prop:norm-zero-ci}.
\end{proof}

\begin{figure}[t]
  \centering
  \includegraphics[scale=0.32]{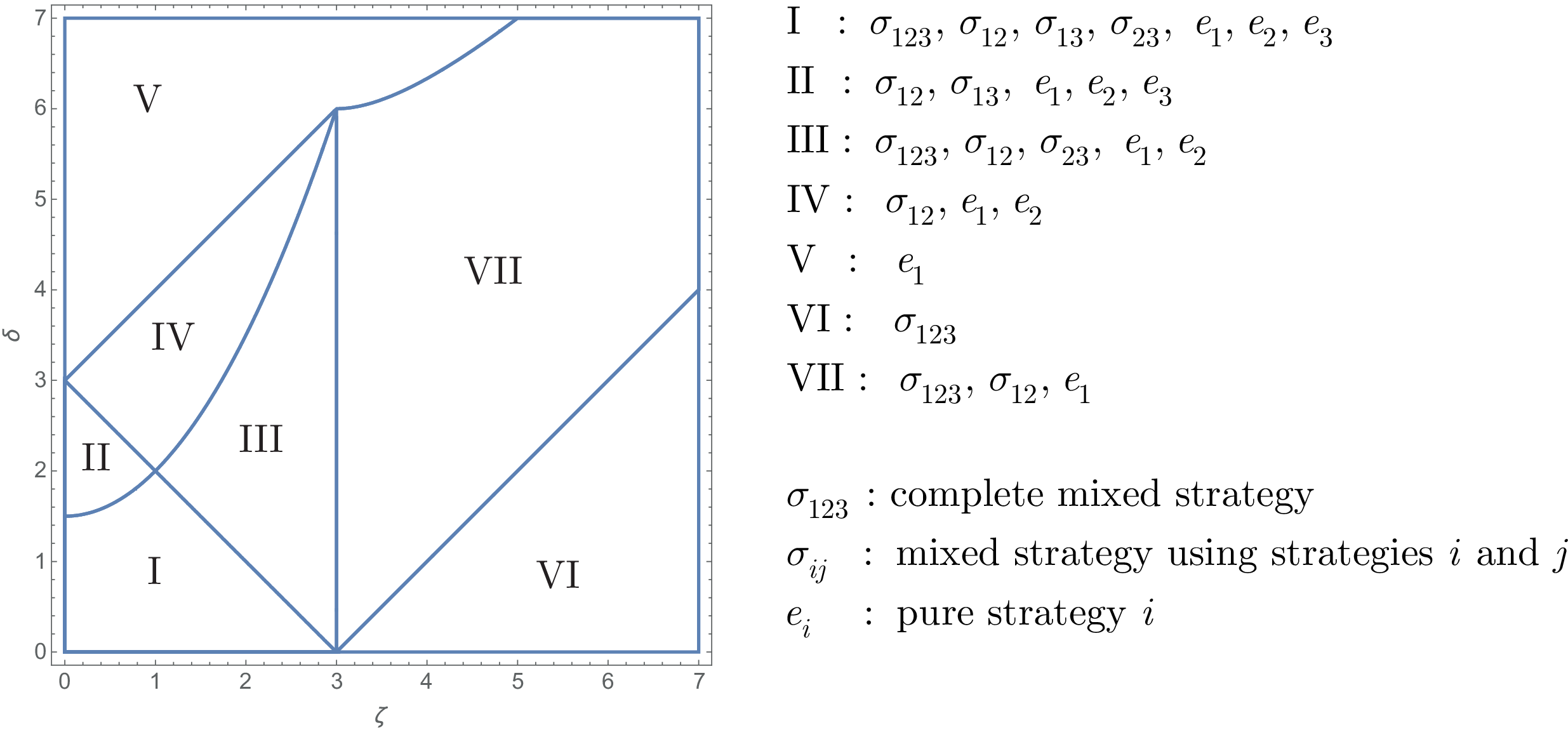}
  \caption[]{\textbf{Nash equilibria under changes in the payoffs of component games.} We find Nash equilibria for the game defined in \eqref{eq:pertub}. In the figure, each region shows the values of $\zeta$ and $\delta$ under which Nash equilibria are found. The horizontal axis: $\zeta$. The vertical axis: $\delta$.

  }
  \label{fig:perturb}
\end{figure}

The bottom line of Figure \ref{fig:dep-eq} presents the decomposition of the symmetric game in Table \ref{tab:1}. In the middle line of Figure \ref{fig:dep-eq}, we show the Nash equilibria of the original game and the component games in the simplexes. In the top line we show the potential function for the identical interest normalized game in $\mathcal{I} \cap \mathcal{N}$, the function $\Phi$ for the zero-sum normalized game in $\mathcal{Z} \cap \mathcal{N}$ and the function $\Phi$ for the zero-sum equivalent potential game in $\mathcal{B}$. The decomposition of the game illustrates how each of the Nash equilibria of the original game is related to the Nash equilibria of component games. For example, the existence of the completely mixed strategy Nash equilibrium for the original game is related to the existence of the zero-sum normalized or identical interest games. Similarly, the existence of the pure strategy Nash equilibria, $(0,1,0)$ and $(0,0,1)$, is due to the existence of the identical interest normalized component.

To illustrate this relationship more explicitly, we present Figure \ref{fig:perturb} in which Nash equilibria are computed under various values of $\zeta$ and $\delta$ for a symmetric game, defined by
    \begin{equation} \label{eq:pertub}
        \begin{pmatrix}
          2 & -1 & -1 \\
          -1 & 2 & -1 \\
          -1 & -1 & 2
        \end{pmatrix}
     +  \zeta \begin{pmatrix}
          0 & -1 & 1 \\
          1 & 0 & -1 \\
          -1 & 1 & 0
        \end{pmatrix}
      +  \delta \begin{pmatrix}
          1 & 1 & 1 \\
          0 & 0 & 0 \\
          0 & 0 & 0
        \end{pmatrix}.
  \end{equation}
Here, a two-player symmetric game $f=(f^{(1)},f^{(2)})$ satisfies $f^{(1)}(x,y) = f^{(2)}(y,x)$ and a matrix can specify a bi-matrix game in Figure \ref{fig:dep-eq}. Note that when $\zeta=1$ and $\delta=1$, the game in \eqref{eq:pertub} becomes the finite game presented in Figure \ref{fig:dep-eq}. Figure \ref{fig:perturb} shows that when a game is close to the identical interest normalized game, equilibrium properties---the numbers of pure strategy and mixed strategy Nash equilibria---of this game are the same as the identical interest normalized game (Region I). Also, when the effect of the zero-sum normalized game, $\zeta$, is sufficiently strong, the corresponding game admits a unique uniform mixed strategy (the property of zero-sum normalized games; Region VI). Similarly, in the case where the effect of a zero-sum equivalent potential game is prevalent, the corresponding game has a unique dominant strategy (the property of zero-sum equivalent potential games; Region V).

Next, motivated by Figure \ref{fig:perturb}, we establish a more precise relationship between Nash equilibria of a given game and those of its component games. If we consider two-player symmetric games with $l$ strategies, then a game can be succinctly identified with an $l \times l$ matrix as in \eqref{eq:pertub}. Then we can explicitly find basis games for subspaces of $\mathcal{I} \cap \mathcal{N}$, $\mathcal{Z} \cap \mathcal{N}$ and $\mathcal{B}$ as follows.
We first define games $\{ S^{(ij)} \}_{ij}$ for identical interest normalized games and games $\{ Z^{(ij)} \}_{ij}$ for zero-sum normalized games:
\[
   S^{(ij)}_{k k'} =
   \begin{cases}
    1 & \text{ if } (k,k')=(i,i) \text{ or } (j,j)\\
    -1 &  \text{ if } (k,k')=(i,j)\text{ or } (j,i) \\
    0 & \text{ otherwise},
   \end{cases}
   \qquad
   Z^{(ij)}_{k k'} = \begin{cases}
                       -1 & \mbox{if } (k,k')=(1,i), (i,j),\text{ or } (j,1) \\
                       1  & \mbox{if } (k,k')=(1,j), (i,1),\text{ or } (j,i) \\
                       0  & \mbox{otherwise}.
                     \end{cases}
\]
Note that $Z^{(ij)}$ is a Rock-Paper-Scissor game involving strategies 1, $i$ and $j$ (see Table \ref{tab:basis}) and, as mentioned, the condition  for  a two player symmetric game to be a potential game can be obtained by the requirement that a game is orthogonal to all $Z^{(ij)}$'s (see \citet{HandR11}).
We also define the games, $D^{(i)}$ and $E^{(i)}$, for $\mathcal{B}$:
\[
    D^{(i)}_{kk'} = \begin{cases}
                      1, & \mbox{if } k=i \\
                      0, & \mbox{otherwise}
                    \end{cases}
    ,\qquad
    E^{(i)}_{kk'} = \begin{cases}
                      1, & \mbox{if } k'=i \\
                      0, & \mbox{otherwise}.
                    \end{cases}
\]
Then, we obtain the following bases for each subspace.
\begin{lem} \label{lem:basis}
We have the following results:\\
  (i) The set of games $\{S^{(ij)}\}_{i=1,\cdots, l, j>i}$ forms a basis for $\mathcal{I} \cap \mathcal{N}$ \\
  (ii) The set of games $\{Z^{(ij)}\}_{i=2,\cdots, l, j>i}$ forms a basis for $\mathcal{Z} \cap \mathcal{N}$ \\
  (iii) The set of games $\{D^{(i)}\}_{i=1,\cdots, l-1}, \{E^{(i)}\}_{i=1,\cdots, l}$ forms a basis for $\mathcal{B}$
\end{lem}
\begin{proof}
 See Lemma \ref{appen:lem-c1}. 
\end{proof}
The proof of Lemma \ref{lem:basis} involves counting the dimension of each subspace and checking whether basis games are independent. Thus, for a given $G$, by our decomposition results, there exists a unique set of coefficients $\{\gamma_{ij} \}, \{\zeta_{ij}\}, \{\delta_i \}$ and $\{\eta_i\}$ such that
\begin{equation}\label{eq:decomp-rep}
  G= \underbrace{\sum_{i=1}^l \sum_{j=i+1}^l \gamma_{ij} S^{(ij)} }_{=: S}+ \underbrace{\sum_{i=2}^l \sum_{j=i+1}^l \zeta_{ij} Z^{(ij)}}_{=: Z}  + \underbrace{\sum_{i=1}^{l-1}\delta_{i} D^{(i)}}_{=: D}+ \sum_{i=1}^{l}\eta_{i} E^{(i)}
\end{equation}
and $G$ is strategically equivalent to $S+Z+D$: $G \sim S + Z + D$.
Denote by $\#(G)$, the number of Nash equilibria of $G$. First note that if $G$ is a symmetric game, then $\#(G) \leq  2^l-1$ (See Lemma 2 (d) and the Theorem in Quint \& Shubik (2002)).
Let
\[
    \underline \gamma := \min_k \min_{j \neq k} \gamma_{kj}, \,\,\, \bar \delta = \max_i \delta_i , \,\,\, \bar \zeta = \frac{(l-1)(l-2)}{2} \max_{j>i} |\zeta_{ij}|
\]
The following proposition identifies conditions under which the equilibrium property of an identical interest normalized component game determines the equilibrium property of the original game.

\begin{prop} \label{prop:number}
  Consider the decomposition in \eqref{eq:decomp-rep}.  \\
  (i) Suppose that $\gamma_{ij}<0$ for all $i,j$. Then $\#(G) = 1$. \com{In addition, if $D=\mathbf{0}$, then the uniform mixed strategy is a unique NE.} \\
  (ii) Suppose that $\gamma_{ij}>0$ for all $i,j$. Suppose that $\delta_i\geq 0$ for all $i$ and $\underline \gamma > \bar \delta + \bar \zeta$. Then $\#(G) = 2^l-1$.
\end{prop}
\begin{proof}
 See Proposition \ref{appen:prop:app1}.
\end{proof}
The result (i) in Proposition \ref{prop:number} is partially known in the literature. \citet{Hofbauer09} define a  class of games, strictly stable games, and show that a strictly stable game possesses a unique mixed strategy Nash equilibrium (For the definition of strictly stable games, see Appendix \ref{appen:app}.). In addition, Proposition \ref{prop:number} (i) shows that a sufficient condition for the strict stability of a game is given by the negative coefficients assigned to identical interest normalized basis games in decomposition. Proposition \ref{prop:number} (ii) also identifies the conditions under which the equilibrium properties of identical interest normalized component game determine equilibrium property of the original game. \citet{Kandori98} shows that a class of coordination games---games satisfying the total bandwagon property---possesses $2^l-1$ number of Nash equilibria. Proposition \ref{prop:number} (ii) shows that (1) the identical interest normalized games with the positive coefficients assigned to basis games satisfy the total bandwagon property and (2)  a game sufficiently close to an identical interest normalized game also satisfies the total bandwagon property. In other words, Proposition \ref{prop:number} (ii) shows how decomposition can be used to find a class of coordination games which satisfies the total bandwagon property.

\com{Finally, we point out that the coefficients, $\gamma_{ij}$, $\zeta_{ij}$, are easily computed by the orthogonal projections in \eqref{eq:orth-proj}, using the relationship
\[
    (\mathcal{ G S} G)_{ij} = \gamma_{ij},\,\, (\mathcal{ G A} G)_{ij} = - \zeta_{ij}
\]
for $i \neq j$ (see Appendix OOO).
}

\tocless \subsection{Contest games}
In this section, we study a class of games, called contest games, which include Tullok contests and all-pay auctions as special cases \citep{Konrad07}.  A contest game is an $n$-player game in which payoff functions are given by
\begin{equation}
f^{(i)}(s)=p^{(i)}(s_{i},s_{-i})v-c_{i}(s_{i})\,\,\text{for \ensuremath{i=1,\cdots,n}},\label{eq:contest-model}
\end{equation}
where $\sum_{i}p^{(i)}(s)=1$ and $p^{(i)}(s)\ge0$ for all $s \geq 0$,
$v>0$, $c_{i}(0)=0$ and $c_{i}(\cdot)$ is continuous, increasing and convex. Then, it is easy to verify that a contest game is a zero-sum equivalent game, since
    \begin{equation}\label{eq:w-fun}
          f^{(i)} \sim  (p^{(i)}(s_{i},s_{-i})-\frac{1}{n})v-\frac{1}{n-1}\sum_{j\ne i}(c_{i}(s_{i})-c_{j}(s_{j})) = : w^{(i)}(s_i, s_{-i}).
    \end{equation}
Thus, Proposition \ref{prop:zero-convex} may be applicable for equilibrium analysis.
The existing literature extensively studies the uniqueness of Nash equilibria under specific assumptions about $p^{(i)}$ (e.g., \citet{HillmanRiley89}, \citet{SO-contest97}). When $p^{(i)}(s)$ is given by
\begin{equation}\label{eq:r-s-games}
  p^{(i)}(s) = \begin{cases}
                \frac{s_i}{\sum_l s_l}, & \mbox{if } s_l > 0 \text{ for some } l  \\
                \frac{1}{n}, & \mbox{if } s_l = 0 \text{ for all } l,
              \end{cases}
\end{equation}
the contest game is called a rent-seeking game. Using Proposition \ref{prop:zero-convex},  we present a completely different way of showing uniqueness of Nash equilibria for rent-seeking games from the existing literature (\citet{SO-contest97}, \citet{Cornes05}). While the existing literature shows the uniqueness of Nash equilibria by using the property of aggregative games, we establish the uniqueness by using the strategic equivalence of the rent-seeking game to a zero-sum game ($w$) and the convexity of $w$. Thus, our approach  shows that recognizing zero-sum equivalent games via decomposition facilitates equilibrium analysis.

Though we adopt a rather simplifying assumption that $c_i(s_i)$ is linear, i.e.,
\[
c_i(s_i) = c_i s_i,
\]
where $c_i >0$ for all $i$, we believe that our method can be extended to the case where $c_i(\cdot)$ is convex (or non-linear) or to other classes of zero-sum equivalent games. We also set $v=1$ for simplicity.

The idea of showing the uniqueness of Nash equilibria is as follows. First, we show that $w^{(i)}(s_i, s_{-i})$ in \eqref{eq:w-fun} is convex in $s_{-i}$ and, thus, the set of Nash equilibria is convex (Lemma \ref{lem:contest1}). Second, we show that when we study the Nash equilibrium of the contest games using the $\Phi_f$ function, it is enough to examine the $\Phi_f$ function defined over the set of players whose action levels are strictly positive, namely the active player set, denoted by $P$, where $P \subset \{1,\cdots, n \}$ and $|P| \geq 2$ (Lemma \ref{lem:contest2}). Then we show that the $\Phi_f$ function defined over $P$ is strictly convex (Lemma \ref{lem:contest3}). This implies that each $P$ set admits at most one Nash equilibrium (Lemma \ref{lem:contest4}). That is, if we define $S(P)$
\[
    S(P):= \{ (s_1, \cdots, s_n) : s_i >0 \text{ for } i \in P \text{ and } s_j =0 \text{ for }  j \not \in P \},
\]
then each $S(P)$ contains at most one Nash equilibrium.  Finally, if there exist two different Nash equilibria, $s^*$ and $t^*$, such that $s^* \in S(P)$ and $t^* \in S(P')$ and $P \neq P'$, then the convexity of Nash equilibria implies that there are infinitely many Nash equilibria, which contradicts the fact that each distinctive $P$ can admit at most one Nash equilibrium. Thus, there exists a unique Nash equilibrium for rent-seeking games defined by \eqref{eq:r-s-games}(Proposition \ref{prop:contest}). We provide all detailed steps and lemmas in Appendix \ref{appen:app}.
\bigskip

\begin{prop} \label{prop:contest}
    The Nash equilibrium for a rent-seeking game defined in \eqref{eq:r-s-games} is unique.
\end{prop}
\begin{proof}
  See Proposition \ref{appen:prop-app2}.
\end{proof}

\tocless \section{Conclusion \label{sec:con}}

In this study, we developed decomposition methods for classes
of games such as zero-sum equivalent games, zero-sum equivalent potential games, zero-sum normalized games, and identical interest normalized games. Our methods rely on the properties of commuting projections in the vector space of games, and identifications of subspaces of games by the ranges and kernels of these projections. The identifications are based on the characterizations for various classes of games.

Next, we showed that two-player finite zero-sum equivalent games have a unique Nash equilibrium. We then studied the class of zero-sum equivalent potential games and showed that two-player finite zero-sum equivalent potential games have generically a unique strictly dominant Nash equilibrium.  We also showed that identical interest normalized games and zero-sum normalized games have a uniform mixed strategy Nash equilibrium. Based on these, we provide two specific applications. In the first application, we demonstrate that decomposition can single out the effect of component games on the Nash equilibrium of the original game. In the second application, the uniqueness of Nash equilibria for rent-seeking games is shown based on the special property of zero-sum equivalent games.

\bigskip

\bigskip

\appendix
\section*{\large{Appendix: for publication}}
\com{
\tableofcontents
\com{
\noindent\textbf{\Large{}Appendix}{\Large \par} \medskip
}

\newpage
}
\renewcommand{\thesection}{A}
\section{Decomposition results\label{appen:decomp_games} }

\subsection{Decomposition via projections\label{appen:projection_games} }

Let $V$ be a vector space. We say that $V$ is the direct sum of $V_1$ and $V_2$ and write $V=V_1 \oplus V_2$ if any $x \in V$ can be written uniquely as $x=x_1+x_2$ with $x_i\in V_i$, $i=1,2$.
Recall that a linear map $P:V \rightarrow V$ is a projection if $P^2 = P$. Then $I-P$ is also a projection and from writing $x =Px + (I- P)x$ we obtain the direct sum decomposition:
\begin{equation}\label{eq:direct-sum}
  V=R(P) \oplus K(P) = R(P) \oplus R(I-P) = K(I-P) \oplus K(P)\,,
\end{equation}
where $R(P)$ and $K(P)$ are the range and kernel of the map $P$, respectively.
Note the following elementary properties for projections.
\begin{lem} \label{lem:vec-proj} Let $P_1$ and $P_2$ be two commuting projections on the vector space $V$. \\
(i) Then $P_1P_2$ is a projection
and
\begin{eqnarray}
R(P_1 P_2) &=& R(P_1) \cap R(P_2)\,, \label{P12R} \\
K(P_1 P_2) &=& K(P_1) + K(P_2) \,. \label{P12K}
\end{eqnarray}
%
(ii) $P_1 + P_2$ is a projection if and only if $P_1P_2=0$, in which case  we have
 \begin{eqnarray}
R(P_1+P_2)=R(P_1) \oplus R(P_2) \label{P1+2R} \,, \\
K(P_1+P_2)=K(P_1)\cap K(P_2) \,. \label{P1+2K}
\end{eqnarray}
%
\end{lem}
\begin{proof}
Some of results are presented, for example, in Chapter 9 in \citet{Kreyszig89}. For completeness, we present the complete proof of all statements. \\
(i) We have $(P_1 P_2)^2=P_1^2 P_2^2 =P_1 P_2$ so $P_1P_2$ is a projection.

\noindent
To prove \eqref{P12R} note that if $x \in R(P_1P_2)$ then $x = P_1P_2y = P_2P_1 y$ and thus $x \in R(P_1)\cap R(P_2)$. Conversely if $x \in R(P_1)\ \cap R(P_2)$ then $x = P_1 y =P_2z$ and thus $x =P_1y =
P_1 P_2y +  P_1(I-P_2)y  = P_1P_2 y + (I-P_2)P_2 z = P_1P_2y$ and thus $x \in R(P_1P_2)$.

\noindent
To prove \eqref{P12K} note that if  $x \in K(P_1P_2)$ then $P_1 P_2x=0$ and then $x = P_2x +(I-P_2)x \in K(P_1) + K(P_2)$  since $P_1 P_2x=0$ and $P_2(I-P_2)x=0$.
Conversely if $x \in K(P_1) + K(P_2)$ then $x = y+z$ with $P_1y=P_2z=0$
and thus $P_1 P_2x = P_2 P_1 y + P_1 P_2 z=0$.

\noindent
(ii) Since $(P_1+P_2)^2 = P_1 + P_2 + 2P_1P_2$, $P_1+P_2$ is a projection if and only if $P_1 P_2=0$. The statements in (ii) by applying (i) to the projections $I-P_1$ and $I-P_2$ since $P_1P_2=0$ implies that $(I-P_1)(I-P_2)= I-(P_1+P_2)$. More specifically, to prove \eqref{P1+2R} we use \eqref{P12K} to obtain
\begin{eqnarray}
R(P_1+P_2)&=&K\left( (1-P_1)(I-P_2)\right)
= K( I-P_1)+ K( I-P_2 )
\nonumber \\
 &=& R(P_1)+ R(P_2) = R(P_1) \oplus R(P_2) \,,
\end{eqnarray}
where in the last equality we used that by \eqref{P12R}  $R(P_1)\cap R(P_2)=R(P_1P_2)= \{0\}$. To prove \eqref{P1+2K} we use \eqref{P12R} and obtain
\begin{eqnarray}
K(P_1+P_2)&=&  R\left( (1-P_1)(I-P_2)\right)=R(I-P_1) \cap R(I-P_2) \nonumber \\
&=& K( P_1) \cap  K( P_2 ) \,.
\end{eqnarray}
\end{proof}

%
%
%
%

From Lemma \ref{lem:vec-proj} we obtain the following vector decompositions.
\begin{prop} \label{appen:vec-decomp}
Let $P_1$ and $P_2$ be commuting projections such that $P_1(I- P_2) =0$.
Then we have
\begin{align}
     \textbf{D1}  \quad V  = & R(P_1) \oplus K(P_1)  \label{eq:vec-d1}\,,\\
     \textbf{D2}  \quad V  = & R(P_2) \oplus K(P_2) \label{eq:vec-d2} \,,\\
     \textbf{D3}  \quad V  = & R(P_1) \oplus K(P_2) \oplus (K(P_1) \cap R(P_2)) \label{eq:vec-d3}\,.
\end{align}
\end{prop}
\begin{proof}
The decompositions \eqref{eq:vec-d1} and  \eqref{eq:vec-d1} are immediate from \eqref{eq:direct-sum}.  For \eqref{eq:vec-d3} note that $P_1$ commute with $I-P_2$ and thus by Lemma \ref{lem:vec-proj} $P_1 + I-P_2$ is a projection. From  \eqref{P1+2R} and \eqref{P1+2K} we obtain
\begin{align*}
   &  K(P_1+ I - P_2) = K(P_1) \cap K(I-P_2) =  K(P_1) \cap R(P_2)\\
   & R(P_1 + I - P_2) = R(P_1) \oplus R(I-P_2) = R(P_1) \oplus K(P_2)
\end{align*}
from Lemma \ref{lem:vec-proj} (ii) and thus \eqref{eq:vec-d3} by \eqref{eq:direct-sum}.
\end{proof}

\subsection{Game-Theoretic Applications}

 Let $f:=(f^{(1)}, \cdots, f^{(n)}): S \rightarrow \mathbb{R}^n$
 where $f^{(i)}:S \rightarrow \mathbb{R}$ for $i=1,\cdots,n$.
We let
\begin{equation}
    \left\Vert f^{(i)} \right\Vert_1 :=\int\left\vert f^{(i)}\right\vert dm, \,\,\, \quad
   \left\Vert f\right\Vert :=\sum_{i=1}^{n}\int\left\vert f^{(i)}\right\vert dm\,. \,\,\,
\end{equation}
We let $L(S,\mathbb{R};m)=\left\{ u:S\to\mathbb{R}\,;\,u\mbox{ is measurable and }\|u\|_1 < \infty\right\}$
and consider the space of games payoffs given by the following vector space:
\begin{align*}
  \mathcal{L}:= & L(S,\mathbb{R}^{n};m)=\left\{ f:S\to\mathbb{R}^{n}\,;\,f\mbox{ is measurable and }\|f\| < \infty\right\} \,.
\end{align*}
Note that $f\in\mathcal{L}$ if and only if $f^{(i)}\in L(S,\mathbb{R};m)$, for each $i=1,2,\cdots,n$.
Recall the operator $T_i u= \frac{1}{|S_i|}\int_{S_i} u(s) \, dm_i(s_i)$  given in equation \eqref{eq:def-t}.

\begin{lem} \label{lem:ti}
We have the following results: \\
(i) The operators $T_{i}$ are projections on $L(S,\mathbb{R};m)$. \\
(ii) The projections $T_{i}$ and $T_{j}$ commute for any $i,j$ and any product of the form $T_{i_{1}} \cdots T_{i_{k}} (I-T_{j_1}) \cdots (I-T_{j_l})$
 is a projection on $L(S,\mathbb{R};m)$.
\end{lem}
\begin{proof}
(i) If $u \in L(S, \mathbb{R}, m)$ then by Fubini theorem $u(s)=u(s_i, s_{-i})$ is integrable with respect to $m_i(s_i)$ (for almost every $s_{-i}$) so that  $T_i u$ is well defined and by Fubini theorem again, $T_i u$ is integrable with respect to $\prod_{l \not =i} m_l(s_l)$.  Since $T_i u$ does not depend on $s_i$,  $T_i u$ is integrable with respect to $m_i$ and we have  $T_i^2 u = T_i u$ and we have $|T_iu| \le T_i|u|$.  By Fubini theorem again
\[
\|T_iu\|_1=\int |T_i u| dm = \int  |T_i u| dm_i(s_i) \prod_{l \not =i} dm_i(s_i)  \le \int  T_i |u| dm_i(s_i) \prod_{l \not =i} dm_i(s_i) \,=\, \|u\|_1\,.
\]
and thus, $T_i$ is bounded.
%
%
(ii) Projections $T_i$ and $T_j$ commute by Fubini theorem and thus by Lemma \ref{lem:vec-proj} any product of $T_i$'s and $I-T_j$'s is again a projection.
 \end{proof}

In the next lemma we prove the basic properties of the projections introduced in Section 2, see the operators $\mathbf{S}$ and $\mathbf{P}$
defined in \eqref{eq:symm_opt},  $\mathbf{G}$ defined in \eqref{eq:orth-proj} and $\mathbf{V}$ defined in \eqref{eq:v-proj}.  Also recall the convenient notation from \eqref{eq:u_M}: for any $M \subset N \,=\,\{1,\cdots,n \}$  we set $u_M = \prod_{l \not \in M}T_l \prod_{k \in M}(I-T_k)u$. By Lemma \ref{lem:ti} the map $u \mapsto u_M$ is a projection.

\begin{lem} \label{lem:proj-space}
We have the following results: \\
(i) $\mathbf{S}$, $\mathbf{P}$, and $\mathbf{G}$ are projections. \\
(ii) $\mathbf{S}$ and $\mathbf{G}$ commute and hence $\mathbf{SG}$ is a projection. \\
(iii) $\mathbf{V}$ is a projection. \\
(iv) $\mathbf{V}$ and $\mathbf{P}$ commute and $\mathbf{VP}=0$, hence $\mathbf{V+P}$ is a projection. \\
(v) $\mathbf{SG}$ and $\mathbf{I-V-P}$ commute and  $\mathbf{SG} (\mathbf{I-V-P})=0$
\end{lem}
\begin{proof}
(i) It is easy to check that $\mathbf{S}$ is a projection and for $\mathbf{P}$ and $\mathbf{G}$ this follows from Lemma \ref{lem:ti} since $(\mathbf{P}f)^{(i)}= T_if^{(i)}$  $(\mathbf{G}f)^{(i)}=f_N^{(i)}$    \\
(ii) We have
\[
    (\mathbf{SG} f)^{(i)}=  \frac{1}{n} \sum_{j=1}^{n} \prod_{l=1}^{n} (I-T_l)f^{(j)} =\prod_{l=1}^{n} (I-T_l)\frac{1}{n} \sum_{j=1}^{n} f^{(j)} =(\mathbf{GS}f)^{(i)}\,.
\]
(iii) First observe that
\begin{equation}\label{eq:con}
  (u_M)_{M'} = \begin{cases}
               0, & \mbox{if } M \neq M' \\
               u_M, & \mbox{if } M = M'
             \end{cases} \,.
\end{equation}
Equation \eqref{eq:con} implies that
if
\[
    (h^{(1)}, \cdots, h^{(n)}) = \mathbf{V} f = (\sum_{M \ni 1} \frac{1}{|M|} \sum_{l \in M} f_M^{(l)}, \cdots, \sum_{M \ni n} \frac{1}{|M|} \sum_{l \in M} f_M^{(l)})
\]
then for any $M$ and $i,j \in M$ we have $h^{(i)}_M = h^{(j)}_M = \frac{1}{|M|}\sum_{l \in M} f^{(l)}_M$ which does not depend on $i,j$. Therefore
\[
\frac{1}{|M|} \sum_{ l \in M} h^{(l)}_M = \frac{1}{|M|}\sum_{l \in M} f^{(l)}_M \,.
\]
This implies that $\mathbf{V}^2=\mathbf{V}$ and thus $\mathbf V$ is a projection. \\
%
%
(iv) Note that if $i \in M$ we have
\[ T_i u_M = T_i \prod_{l \notin M} T_l \prod_{k \in M} (I-T_k) u =
\prod_{l \notin M} T_l \prod_{k \in M} (I-T_k) (T_i u) = (T_i u)_M =0
\]
where the last inequality follows from $\displaystyle T_i \prod_{k \in M} (I-T_k) =
T_i(I-T_i)  \prod_{\substack {k \in M \\ k \not=i}} (I-T_k)=0$.
Therefore
\begin{align}
(\mathbf{PV} f)^{(i)} & = T_i (\mathbf{V}f)^{(i)} = T_i \sum_{M \ni i} \frac{1}{|M|} \sum_{j \in M} f_M^{(j)} =0 \\
(\mathbf{VP}f)^{(i)}&= \sum_{M \ni i} \frac{1}{|M|} \sum_{j \in M} (T_jf^{(j)})_M = 0
\end{align}
Thus $\mathbf{PV}=\mathbf{VP}=0$ and by Lemma \ref{lem:vec-proj} $\mathbf{V} + \mathbf{P}$ is a projection.
\\
(v) Note first that the fact $\mathbf{SGP} = \mathbf{PSG}=0$ is proved exactly as in (iv) since $(\mathbf{SG}f)^{(i)} = \frac{1}{|N|}\sum_{j \in N} f^{(j)}_N$.
Moreover, if we let
\[
   (\mathbf{D} f)^{(i)}= \sum_{\substack{M \ni i \\ M \neq N} } \frac{1}{|M|} \sum_{j \in M} f_M^{(j)}
\]
we have $\mathbf{V} = \mathbf{SG} + \mathbf{D}$.  Therefore we have
\begin{equation}\label{eq:sgd}
\mathbf{SG}( \mathbf{I} - \mathbf{V} + \mathbf{P})= \mathbf{SG} \mathbf{D} \quad \textrm{ and }
( \mathbf{I} - \mathbf{V} + \mathbf{P})\mathbf{SG} =\mathbf{D} \mathbf{SG}
\end{equation}
To conclude, recall that $(u_M)_{M'}=0$ if $M \not= M'$ in \eqref{eq:con}. From the
definition of $\mathbf{D}$ we note that $((\mathbf{D}f)^{(i)})_N=0$ and thus
\[
(\mathbf{SG} \mathbf{D} f)^{(i)} = \frac{1}{|N|} \sum_{j \in N} ((\mathbf{D}f)^{(j)})_N \,=\,0\,.
\]
On the other hand we have $((\mathbf{SG}f)^{(i)})_M =0$ if $M \not= N$
and thus
\[
( \mathbf{D} \mathbf{SG} f)^{(i)} \,=\, \sum_{\substack{M \ni i \\ M \neq N} } \frac{1}{|M|} \sum_{j \in M} ((\mathbf{SG}f)^{(i)})_M =0
\]
%
%
%
This proves that $\mathbf{SG} \mathbf{D}= \mathbf{D}\mathbf{SG} =0$ and therefore by \eqref{eq:sgd} (v) holds.
\end{proof}
%
%

\noindent \textbf{Proof of Theorem \ref{thm:main}.}
We let
\[
    P_1 = \mathbf{SG}, \,\,\, P_2 = \mathbf{V + P}
\]
in \eqref{eq:vec-d1}, \eqref{eq:vec-d2}, and \eqref{eq:vec-d3}. From Lemma \ref{lem:proj-space} and Proposition \ref{appen:vec-decomp}, we have the following decomposition:
\begin{alignat}{2}
     \textbf{D1}  \quad \mathcal{L}  = & R(\mathbf{SG}) \oplus K(\mathbf{SG})  \label{eq:d1}  \\
     \textbf{D2}  \quad \mathcal{L}  = & R(\mathbf{V+P}) \oplus K(\mathbf{V+P}) \label{eq:d2} \\
     \textbf{D3}  \quad \mathcal{L}  = & R(\mathbf{SG}) \oplus K(\mathbf{V+P}) \oplus (K(\mathbf{SG}) \cap R(\mathbf{V+P}))   \label{eq:d3}
\end{alignat}
Then from Propositions \ref{appen-prop:iinorm}, \ref{prop:zero-equiv}, \ref{appen-prop:pot-char}, and \ref{prop:zero-norm}, proven below we obtain the decompositions in Theorem \ref{thm:main}
$\square$.

Note that we also have
\begin{equation}\label{eq:char-b}
  K(\mathbf{SG}) \cap R(\mathbf{V+P})= K(\mathbf{SG}) \cap K(\mathbf{I-(V+P)})=K(\mathbf{I-(D+P)})= R(\mathbf{D+P}).
\end{equation}
Next, we will show that each of these ranges and kernels in \eqref{eq:d1}, \eqref{eq:d2}, and \eqref{eq:d3} are indeed equivalent to the corresponding subspaces presented in Theorem \ref{thm:main}.

\begin{prop}[\textbf{Identical interest normalized games}] \label{appen-prop:iinorm}
We have
\[
R(\mathbf{SG}) = \mathcal{I} \cap \mathcal{N}.
\]
\end{prop}
\begin{proof}
Let $f \in R(\mathbf{SG})$. Then $f = \mathbf{SG}f=\mathbf{GS} f$. Thus $f \in \mathcal{I}\cap \mathcal{N}$. Suppose that $f \in \mathcal{I}\cap \mathcal{N}$. Then $f=(h,\cdots, h)$ for some scalar valued function $h$ such that
\[
    \int_{S_i} h(s) d m_i(s_i) =0 \,\,\text{ for all } i.
\]
Thus, $ \mathbf{S} f = f$. Also since  $(I-{T}_i) h = h - \int_i h dm_i = h $ for all $i$, we have $\mathbf{G} f= f$. Thus $\mathbf{SG} f = \mathbf{G} f = f$ and we have $f \in  R(\mathbf{SG})$.  \end{proof}

\begin{prop}[\textbf{Zero-sum equivalent games}] \label{prop:zero-equiv}
We have
\[
    K(\mathbf{SG}) = \mathcal{Z}+\mathcal{E}
\]
\end{prop}
\begin{proof}
We first have the following equivalence: $  f \in K(\mathbf{SG}) \iff \mathbf{SG}  f =0$. Thus we will show that $ \mathbf{SG} f =0 \iff f \in \mathcal{Z} + \mathcal{E}$ or equivalently,
    \[
     \sum_{i=1}^{n}f^{(i)}\in K(\prod_{l=1}^{n}(I-T_{l})) \iff f\in\mathcal{Z}+\mathcal{E}
    \]
\noindent $(\impliedby)$    If $f\in\mathcal{Z}+\mathcal{E}$,
then $f^{(i)}=g^{(i)}+h^{(i)}$, where $\sum_{i=1}^{n}g^{(i)}=0$
and $h^{(i)}\in R(T_{i})$. Therefore, we have
\[
\sum_{i=1}^{n}f^{(i)}=\sum_{i=1}^{n}T_{i}q^{(i)},
\]
for some $q^{(1)},\cdots,q^{(n)}$, and clearly, we have $\left(\prod_{l=1}^{n}(I-T_{l})\right)(\sum_{i=1}^{n}T_{i}q^{(i)})=0$.

\noindent$(\implies)$ Conversely, suppose that $\sum_{i=1}^{n}f^{(i)}\in K(\prod_{l=1}^{n}(I-T_{l}))$.
Then, for each $i$,
\[
f^{(i)}= \prod_{l=1}^{n}(I-T_{l}) f^{(i)} + (I-\prod_{l=1}^{n} (I-T_{l}))  f^{(i)} =:m^{(i)}+n^{(i)}.
\]
Then $ \sum_{i=1}^{n} m^{(i)} = \sum_{i=1}^{n} \prod_{l=1}^n (I-T_l) f^{(i)} =0 $ because $\sum_{i=1}^{n}f^{(i)}\in K(\prod_{l=1}^{n}(I-T_{l}))$. Also since $n^{(i)}\in K(\prod_{l=1}^{n}(I-T_{l}))=R(T_{1})+\cdots+R(T_{n})$,
we have, for each $i$,
\[
n^{(i)}\,=\,\sum_{j=1}^{n}T_{j}n_{j}^{(i)},
\]
for some $\{n_{j}^{(i)}\}_{j=1}^{n}$. In this way, we find $\{n_{j}^{(i)}\}_{i,j}$.
For each $i$, we write
\[
n^{(i)}\,=\,\left(\sum_{j=1}^{n}T_{j}n_{j}^{(i)}-\sum_{j=1}^{n}T_{i}n_{i}^{(j)}\right)+\underbrace{\sum_{j=1}^{n}T_{i}n_{i}^{(j)}}_{\in R{(T_{i})}}\,.
\]
Then, we have
\[
\sum_{i=1}^{n}(\sum_{j=1}^{n}T_{j}n_{j}^{(i)}-\sum_{j=1}^{n}T_{i}n_{i}^{(j)})\,=\,0.
\]
Thus $f=(f^{(1)},\cdots, f^{(n)})$ can be written as
\begin{align*}
    = & \underbrace{(m^{(1)},\cdots, m^{(n)}) + (\sum_{j=1}^{n}T_{j}n_{j}^{(1)}-\sum_{j=1}^{n}T_{1}n_{1}^{(j)}, \cdots, \sum_{j=1}^{n}T_{j}n_{j}^{(n)}-\sum_{j=1}^{n}T_{n}n_{n}^{(j)})}_{\in \mathcal{Z}} \\ + & \underbrace{(\sum_{j=1}^{n}T_{1}n_{1}^{(j)}, \cdots, \sum_{j=1}^{n}T_{n}n_{n}^{(j)})}_{\in \mathcal{E}}
\end{align*}
This shows that $f\in\mathcal{Z}+\mathcal{E}$, and concludes the
proof of the claim. \\
\end{proof}
The following proposition gives detailed  characterizations for potential games.
\begin{prop}[\textbf{Potential games}] \label{appen-prop:pot-char}
The following statements are equivalent. \\
 (i) $f$ is a potential game. \\
 (ii)
 \begin{equation}\label{eq:pot-con}
   (I-T_i)(I-T_j) f^{(i)} =(I-T_i)(I-T_j) f^{(j)} \textrm{ for all } i,j \,.
 \end{equation}
 (iii)
 \begin{equation}\label{eq:appen-pot-con2}
  f^{(i)}_M = \frac{1}{|M|} \sum_{j \in M} f^{(j)}_M \textrm{ and for all }  i \in M\, \textrm{ for all non-empty } M .
\end{equation}
(iv)
\begin{equation}\label{eq:pot-con3}
  f^{(i)} = T_i f^{(i)} +  \sum_{M \ni i} \frac{1}{|M|}  \sum_{j \in M }f^{(j)}_M  \textrm { for all } i.
  \end{equation}

\noindent Thus we have
\[
    R(\mathbf{V+P})=\mathcal{I}+\mathcal{E}.
\]
\end{prop}
\begin{proof}
((i) $\implies$ (ii)) Suppose that $f$ is a potential game. Then for all $i$,  $f^{(i)}(s)= \phi + T_i h^{(i)}$ for some $h^{(i)}$. Thus we have
\[
    (I-T_i)(I-T_j)(f^{(i)}(s)-f^{(j)}(s)) = (I-T_i)(I-T_j)(T_i h^{(i)}(s_{-i})- T_j h^{(j)}(s_{-j}))=0
\]
\noindent ((ii) $\implies$ (iii)) Suppose that condition \eqref{eq:pot-con} holds and let $M \ni \{i,j\}$. Then
\begin{align}
f_M^{(i)}=& \prod_{l \notin M} T_l \prod_{k \in M}(I-T_k) f^{(i)}  \nonumber \\
=& \prod_{l \notin M} T_l \prod_{\substack{k \in M
\\ i,j \notin M} }(I-T_k) (I-T_i)(I-T_j) f^{(i)} \nonumber\\
=& \prod_{l \notin M} T_l \prod_{\substack{k \in M
\\ i,j \notin M} }(I-T_k) (I-T_i)(I-T_j) f^{(j)}
= f_M^{(j)}\nonumber
\end{align}
and therefore $f_M^{(i)}$ is independent of $i$ if $i \in M$ and thus \eqref{eq:appen-pot-con2} holds.\\
\noindent ((iii) $\implies$ (iv)) If \eqref{eq:appen-pot-con2} holds, we have
\begin{align}
f^{(i)} =& T_i f^{(i)} + (I-T_i)f^{(i)} \nonumber\\
=&  T_i f^{(i)} + \sum_{M \ni i} f_M^{(i)} \nonumber\\
=&  T_i f^{(i)} + \sum_{M \ni i} \frac{1}{|M|} \sum_{j\in M} f_M^{(j)} \,. \nonumber
\end{align}
which establishes \eqref{eq:pot-con3}. \\
\noindent ((iv) $\implies$ (i))    Suppose that \eqref{eq:pot-con3} holds. Then
  \[
    f^{(i)} = \sum_{M \ni i} \frac{1}{|M|}\sum_{j \in M }f^{(j)}_M +T_i f^{(i)} =  \underbrace{\sum_{\substack{M \subset N \\ M \neq \emptyset}} \frac{1}{|M|} \sum_{j \in M }f^{(j)}_M}_{\let \scriptstyle \textstyle \substack{ \equiv \phi}} \underbrace{ -\sum_{\substack{M \not \in i \\ M \neq \emptyset}} \frac{1}{|M|} \sum_{j \in M }f^{(j)}_M +T_i f^{(i)}}_{\let \scriptstyle \textstyle \substack{ \equiv h^{(i)} }}
  \]
  where the first term $\phi$ does not depend on $i$
  and the second term $h^{(i)}$ satisfies $h^{(i)}=T_ih^{(i)}$ since $f_M^{(j)}=T_i f_M^{(j)}$ if $ i \notin M$. This shows that $f$ is a potential game.
\end{proof}

\begin{prop}[\textbf{Zero-sum normalized games}] \label{prop:zero-norm}
    We have
    \[
        K(\mathbf{V+P}) =  \mathcal{Z} \cap \mathcal{N}
    \]
\end{prop}
\begin{proof}
 We will show that
 $f$ is a zero-sum normalized game if and only if
    \begin{equation}\label{eq:char-zn}
      \sum_{M \ni i } \frac{1}{|M|} \sum_{j \in M} f_M^{(j)} + T_i f^{(i)} =0
    \end{equation}
    for all $i$.
  Suppose that $f$ is a zero-sum normalized game. Let $i$ be fixed. Then since $f$ is normalized, $T_i f^{(i)}=0$. Also if $j \not \in M$, then $f^{(j)}_M=0$, again since $f$ is normalized. Thus if $f$ is normalized, then
  \[
    \sum_{j \not \in M} f_M^{(j)} = 0.
  \]
  Thus we find that
  \[
    \sum_{M \ni i } \frac{1}{|M|} \sum_{j \in M} f_M^{(j)} = \sum_{M \ni i } \frac{1}{|M|} \sum_{j =1}^n f_M^{(j)} =0
  \]
  since $f$ is a zero-sum game. \\
  Now suppose that \eqref{eq:char-zn} holds. Let $i$ be fixed and $M' \ni i$. Applying $\prod_{l \not \in M'} T_l \prod_{k \in M'} (I-T_k) $ at \eqref{eq:char-zn}, from  \eqref{eq:con}, we find that
  \[
    \frac{1}{|M'|} \sum_{j \in M'} f_{M'}^{(j)} =0
  \]
  Thus for all $M' \ni i$, we have $\frac{1}{|M'|} \sum_{j \in M'} f_{M'}^{(j)} =0$. Thus $T_i f^{(i)}=0$. That is, $f$ is normalized. By varying $i$, we also find that $\sum_{j \in M} f_M^{(j)}=0$ for all $M \neq \emptyset $. Also note that since $f$ is normalized, $(\sum_{j=1}^{n} f^{(j)})_{\emptyset} =0$ and $\sum_{j \notin M}  f_M^{(j)} =0$.
  Thus, we find that
  \[
        \sum_{M \subset N}(\sum_{j=1}^{n} f^{(j)})_M=\sum_{\substack{M \subset N \\ M \neq \emptyset} }(\sum_{j=1}^{n} f^{(j)}_M) =\sum_{\substack{M \subset N \\ M \neq \emptyset} }(\sum_{j \in M } f^{(j)}_M) =0
  \]
  Thus we have $\sum_{j=1}^{n} f^{(j)}=0$.
\end{proof}

\com{
Also it is easy to check that
\[
    range(\mathcal{GS})=(\mathcal{I} \cap \mathcal{N}), \,\, range(\mathcal{G}(I-\mathcal{S}))=(\mathcal{Z} \cap \mathcal{N} ), \,\,\, range(I- \mathcal{G})= \mathcal{B}
\]
}

\newpage

\bibliographystyle{elsarticle-harv}
\bibliography{evolutionary_games}

\begin{thebibliography}{40}
\expandafter\ifx\csname natexlab\endcsname\relax\def\natexlab#1{#1}\fi
\expandafter\ifx\csname url\endcsname\relax
  \def\url#1{\texttt{#1}}\fi
\expandafter\ifx\csname urlprefix\endcsname\relax\def\urlprefix{URL }\fi

\bibitem[{Alger and Weibull(2013)}]{Alger13}
Alger, I., Weibull, J., 2013. Homo moralis-preference evolution under
  incomplete information and assotative matching. Econometrica 81, 2269--2302.

\bibitem[{Aumann and Sorin(1989)}]{Aumann89}
Aumann, R., Sorin, S., 1989. Cooperation and bounded recall. Games and Economic
  Behavior 1~(1), 5--39.

\bibitem[{Barron(2008)}]{Barron08}
Barron, E.~N., 2008. Game Theory: An Introduction. Wiley.

\bibitem[{Bregman and Fokin(1998)}]{Bregman88}
Bregman, L.~M., Fokin, I.~N., 1998. On separable non-cooperative zero-sum
  games. Optimization 44~(1), 69--84.

\bibitem[{Cai et~al.(2015)Cai, Candogan, Daskalakis, and Papadimitriou}]{Cai15}
Cai, Y., Candogan, O., Daskalakis, C., Papadimitriou, C., 2015. A multiplayer
  generalization of the minmax theorem. Mathematics of Operation Research,
  Forthcoming.

\bibitem[{Candogan et~al.(2011)Candogan, Menache, Ozdaglar, and
  Parrilo}]{Candogan2011}
Candogan, O., Menache, I., Ozdaglar, A., Parrilo, P.~A., 2011. Flows and
  decompositions of games: Harmonic and potential games. Mathemtics of
  Operations Research 36~(3), 474--503.

\bibitem[{Cornes and Hartley(2005)}]{Cornes05}
Cornes, R., Hartley, R., 2005. Asymmetric contests with general technologies.
  Economic Theory 26, 923--946.

\bibitem[{Dasgupta and Maskin(1986)}]{Dasgupta86}
Dasgupta, P., Maskin, E., 1986. The existence of equilibrium in discontinuous
  economic games, i: Theory. Review of Economic Studies 53, 1--56.

\bibitem[{Debreu(1952)}]{Debreu52}
Debreu, G., 1952. A social equilibrium existence theorem. Proceedings of the
  National Academy of Sciences 38, 886--893.

\bibitem[{Duggan(2007)}]{Duggan07}
Duggan, J., 2007. Equilibrium existence for zero-sum games and spatial models
  of elections. Games and Economic Behavior 60, 52--74.

\bibitem[{Fan(1952)}]{Fan52}
Fan, K., 1952. Fixed point and minimax theorems in local convex topological
  linear spaces. Proceedings of the National Academy of Sciences 38, 121--126.

\bibitem[{Glicksberg(1952)}]{Glicksberg52}
Glicksberg, I.~L., 1952. A further generalization of the {K}akutani fixed point
  theorem with application to {N}ash equilibrium points. Proceedings of the
  American Mathematical Society 38, 170--174.

\bibitem[{Hillman and Riley(1989)}]{HillmanRiley89}
Hillman, A.~L., Riley, J.~G., 1989. Political contestable rents and transfers.
  Economics and Politics 1, 17--39.

\bibitem[{Hofbauer and Sandholm(2009)}]{Hofbauer09}
Hofbauer, J., Sandholm, W., 2009. Stable games and their dynamics. Journal of
  Economic Theory 144, 1665--1693.

\bibitem[{Hofbauer and Schlag(2000)}]{Hofbauer00}
Hofbauer, J., Schlag, K., 2000. Sophisticated imitation in cyclic games.
  Journal of Evolutionary Economics.

\bibitem[{Hwang and Rey-Bellet(2011)}]{HandR11}
Hwang, S.-H., Rey-Bellet, L., 2011. Decomposition of two player games:
  Potential, zero-sum, and stable games. arXiv:1106.3552.

\bibitem[{Hwang and Rey-Bellet(2016)}]{HandRTest15}
Hwang, S.-H., Rey-Bellet, L., 2016. Simple characterization of potential games
  and zero-sum games, arxiv.org:1602.04410.

\bibitem[{Kalai and Kalai(2010)}]{Kalai10}
Kalai, A., Kalai, E., 2010. Cooperation and competition in strategic games with
  private information. EC '10 Proceedings of the 11th ACM conference on
  Electronic commerce.

\bibitem[{Kandori and Rob(1998)}]{Kandori98}
Kandori, M., Rob, R., 1998. Bandwagon effects and long run technology choice.
  Games and Ecoomic Behavior 22, 30--60.

\bibitem[{Konrad(2009)}]{Konrad07}
Konrad, K.~A., 2009. Strategy and Dynamics in Contests. Oxford University
  Press, Oxford.

\bibitem[{Kreyszig(1989)}]{Kreyszig89}
Kreyszig, E., 1989. Introductory Functional Analysis with Applications. Wiley.

\bibitem[{McLennan and Park(1999)}]{MacLennan&Park99}
McLennan, A., Park, I.-U., 1999. Generic 4x4 two person games have at most 15
  {N}ash equilibria. Games and Economic Behavior 26, 111--130.

\bibitem[{Monderer and Shapley(1996)}]{Monderer96}
Monderer, D., Shapley, L.~S., 1996. Potential games. Games and Economic
  Behavior 14, 124--143.

\bibitem[{Morris and Ui(2004)}]{Morris04}
Morris, S., Ui, T., 2004. Best response equivalence. Games and Economic
  Behavior 49, 260--287.

\bibitem[{Myerson(1997)}]{Myerson97}
Myerson, R.~B., 1997. Dual reduction and elementary games. Games and Economic
  Behavior 21, 183--202.

\bibitem[{Nikaido and Isoda(1955)}]{Nikaido55}
Nikaido, H., Isoda, K., 1955. Note on non-cooperative convex games. Pacific
  Journal of Mathematics 5, 807--815.

\bibitem[{Quint and Shubik(1997)}]{Quint&Shubik97}
Quint, T., Shubik, M., 1997. A theorem on the number of {N}ash equilibria in a
  bimatrix game. Internation Journal of Game Theory 26, 353--359.

\bibitem[{Quint and Shubik(2002)}]{Quint&Shubik02}
Quint, T., Shubik, M., 2002. A bound on the number of {N}ash equilibria in a
  coordination game. Economics Letters 77, 323--327.

\bibitem[{Reny(2003)}]{Reny03New}
Reny, P.~J., 2003. On the existence of pure and mixed strategy {N}ash
  equilibria in discontinuous games. Econometrica 67~(5), 1029--1056.

\bibitem[{Rosen(1965)}]{Rosen65}
Rosen, J.~B., 1965. Existence and uniqueness of equilibrium points for concave
  $n$-person games. Econometrica 33~(3), 520--534.

\bibitem[{Sandholm(2010{\natexlab{a}})}]{Sandholm10}
Sandholm, W., 2010{\natexlab{a}}. Decompositions and potentials for normal form
  games. Games and Economic Behavior 70, 446--456.

\bibitem[{Sandholm(2010{\natexlab{b}})}]{Sandholm10Book}
Sandholm, W., 2010{\natexlab{b}}. Population Games and Evolutionary Dynamics.
  MIT Press.

\bibitem[{Savani and von Stengel(2006)}]{Savani06}
Savani, R., von Stengel, B., 2006. Hard-to-solve bimatrix games. Econometrica
  74~(2), 397--429.

\bibitem[{Shapley(1974)}]{Shapley74}
Shapley, L.~S., 1974. A note on the {L}emke-{H}owson algorithm. Mathematical
  Programming Study 1, 175--89.

\bibitem[{Szidarovszky and Okuguchi(1997)}]{SO-contest97}
Szidarovszky, F., Okuguchi, K., 1997. On the existence and uniqueness of
  equilibrium in rent-seeking games. Games and Economic Behavior 18, 135--140.

\bibitem[{Ui(2000)}]{Ui00}
Ui, T., 2000. A {S}hapley value representation of potential games. Games and
  Economic Behavior 31, 121--135.

\bibitem[{Weibull(1995)}]{Weibull95}
Weibull, J., 1995. Evolutionary Games Theory. Cambridge, MA.

\bibitem[{Wilson(1971)}]{Wilson71}
Wilson, R., 1971. Computing equilibria of $n$-person games. SIAM Journal of
  Applied Mathematics 21~(1), 80--87.

\bibitem[{Wu and Jiang(1962)}]{WuandJiang52}
Wu, W.-T., Jiang, J.-H., 1962. Essential equilibrium of points of $n$-person
  non-cooperative games. Scientia Sinica 11, 1307--1322.

\bibitem[{Zeeman(1980)}]{Zeeman80}
Zeeman, E., 1980. Population dynamics from game theory. Global Theory of
  Dynamical Systems 819.

\end{thebibliography}

\newpage

\newpage
\section*{\large{Appendix: Only for online publication}}

\renewcommand{\thesection}{B}
\section{Other proofs}\label{appen:other-proofs}


\begin{proof}[\textbf{Proof of Proposition \ref{prop:zero-convex}.}]\label{proof:zero-convex}
 We show (ii)((i) follows similarly).
Let $i$ and $s_{i}$ be fixed. Then from the discussion before the proposition, $w^{(i)}$ is concave in $s_i$ for all $i$. Thus there exists a Nash equilibrium. We next show that $\Phi(s)=\sum_{i}\max_{s_{i}\in S_{i}}w^{(i)}(s_{i},s_{-i})$
is strictly convex. Let $t',$ $t''\in S$ be given. Then $u',$ $u''\in S$
be given such that $w^{^{(i)}}(u_{i}',t_{-i}')=\max_{s_{i}\in S_{i}}w^{(i)}(s_{i},t_{-i}')$
and $w^{^{(i)}}(u_{i}'',t_{-i}'')=\max_{s_{i}\in S_{i}}w^{(i)}(s_{i},t_{-i}'')$
for all $i$. Let $\alpha\in(0,1)$ and $t^{*}$ be such that $w^{^{(i)}}(t_{i}^{*},((1-\alpha)t'+\alpha t'')_{-i})=\max_{s_{i}\in S_{i}}w^{(i)}(s_{i},((1-\alpha) t'+ \alpha t'')_{-i})$
for all $i$. Then we have
\begin{align*}
& (1-\alpha)\Phi(t')+\alpha \Phi(t'') =(1-\alpha)\sum_{i}w^{(i)}(u_{i}',t_{-i}')+\alpha \sum_{i}w^{^{(i)}}(u_{i}'',t_{-i}'')\\
 & \ge(1-\alpha)\sum_{i}w^{(i)}(t_{i}^{*},t_{-i}')+\alpha \sum_{i}w^{^{(i)}}(t_{i}^{*},t_{-i}'') >\sum_{i}w^{(i)}(t_{i}^{*},(1-\alpha)t'_{-i}+\alpha t''_{-i})\\
 & =\Phi((1-\alpha)t'+\alpha t'').
\end{align*}
Thus $\Phi_f(s)$ is strictly convex and the minimizer of $\Phi_f$ is unique. Since the Nash equilibrium is a minimizer of $\Phi_f$, the Nash equilibrium is unique.

\end{proof}

\begin{proof}[\textbf{Proof of Corollary \ref{cor:2p-zero-finite}.}]\label{proof:2p-zero-finite}
Let $f=w+h$, where $h$ is a non-strategic game. Then, $w^{(1)}(\sigma_{1},\sigma_{2})$
is convex in $\sigma_{2}$ and $w^{(2)}(\sigma_{1},\sigma_{2})$ is
convex in $\sigma_{1}.$ By Proposition \ref{prop:zero-convex}, the set of Nash equilibria is convex.
Suppose that $f$ has two distinct  Nash equilibria, $\rho^{*}$ and $\sigma^{*}$, where
$\rho^{*}\ne\sigma^{*}.$ Then, for all $t\in(0,1)$, $(1-t)\rho^{*}+t\sigma^{*}$ is a Nash equilibrium
since the set of Nash equilibria is convex. This contradicts Condition \hyperlink{con-N}{\textbf{(N)}} because of Lemma 2.2 in \citet{Quint&Shubik97}.
\end{proof}

\begin{proof}[\textbf{Proof of Proposition \ref{prop:n-p-pop}}]\label{proof:prop:n-p-pop}

We let
\[
    \mathcal{D} := \{ f \in \mathcal{L} : f^{(i)}(s) := \sum_{l \neq i} \zeta_l (s_{-l}) \text{ for all } i \}
\]

We first show that
\[
\mathcal{B}=(\mathcal{Z}+\mathcal{E})\cap(\mathcal{I}+\mathcal{E})=\mathbf{S}(\mathbf{P}(\mathcal{L}))+\mathcal{E}.
\]
Let $f\in(\mathcal{Z}+\mathcal{E})\cap(\mathcal{I}+\mathcal{E})$.
Then, $f=g_{1}+h_{1}$, for $g_{1}\in\mathcal{Z}$ and$\ h_{1}\in\mathcal{E}$,
and $f=g_{2}+h_{2}$, for $g_{2}\in\mathcal{I}$ and $h_{2}\in\mathcal{E}.$
Thus, we have
\begin{equation}
g_{1}+h_{1}=g_{2}+h_{2},\label{eq:pr}
\end{equation}
and applying $\mathcal{S}$ to (\ref{eq:pr}), we obtain
\[
f=\mathcal{S}(h_{1}-h_{2})+h_{2}.\text{ }
\]
Thus, since $h_{1}-h_{2}\in \mathbf{P} (\mathcal{L})$, $f\in \mathbf{S}(\mathbf{P}(\mathcal{L}))+\mathcal{E}$.
Conversely, let $f\in \mathbf{S}(\mathbf{P}(\mathcal{L}))+\mathcal{E}$. Obviously,
$f\in\mathcal{I}+\mathcal{E}$. In addition, $f=\mathbf{S}(\mathbf{P}(g))+h_{1}$,
for $\ g\in\mathcal{L}$ and $h_{1}\in\mathcal{E}.$ Thus,
\[
f=\mathbf{S}(\mathbf{P}(g))+h_{1}=-(\mathbf{I}-\mathbf{S})(\mathbf P(g))+\mathbf P(g)+h_{1}\in\mathcal{Z}+\mathcal{E}\text{.}
\]
This shows that
\[
(\mathcal{Z}+\mathcal{E})\cap(\mathcal{I}+\mathcal{E})=\mathbf{S}(\mathbf P(\mathcal{L}))+\mathcal{E}\text{.}
\]
Note that
\begin{align*}
\mathbf{S}(\mathbf{P}(\mathcal{L}))+\mathcal{E} & \mathcal{=}\{f:f^{(i)}=\sum_{l=1}^{n} \zeta_{l}(s_{-l})\text{ for some }\{\zeta_{l}\}_{l=1}^{n}\text{ and for all }i\}+\mathcal{E}\\
 & =\{f:f^{(i)}=\sum_{l\neq i}\zeta_{l}(s_{-l})\text{ for some }\{\zeta_{l}\}_{l=1}^{n}\text{ and for all }i\}+\mathcal{E}\\
 & =\mathcal{D}+\mathcal{E}.
\end{align*}

Now observe that
\begin{align*}
 & (\sum_{l\neq1}\zeta_{l}(s_{-l}),\sum_{l\neq2}\zeta_{l}(s_{-l}),\cdots,\sum_{l\neq n}\zeta_{l}(s_{-l}))\\
 & \sim(\sum_{l=1}^{n}\zeta_{l}(s_{-l}),\sum_{l=2}^{n}\zeta_{l}(s_{-l}),\cdots,\sum_{l=1}^{n}\zeta_{l}(s_{-l})).
\end{align*}
Hence, the first result follows from $\mathcal{D}+\mathcal{E}=(\mathcal{I}+\mathcal{E})\cap(\mathcal{Z}+\mathcal{E}$).
For the second result, observe that
\begin{align*}
(\sum_{l\neq1}\zeta_{l},\sum_{l\neq2}\zeta_{l},\cdots,\sum_{l\neq n}\zeta_{l}) & \sim(\sum_{l\neq1}\zeta_{l}-(n-1)\zeta_{1},\sum_{l\neq2}\zeta_{l}-(n-1)\zeta_{2},\cdots,\sum_{l\neq n}\zeta_{l}-(n-1)\zeta_{n})\\
 & =(\sum_{l\neq1}(\zeta_{l}-\zeta_{1}),\sum_{l\neq2}(\zeta_{l}-\zeta_{2}),\cdots,\sum_{l\neq n}(\zeta_{l}-\zeta_{n}))\\
 & =(\sum_{l>1}(\zeta_{l}-\zeta_{1}),\zeta_{1}-\zeta_{2},\zeta_{1}-\zeta_{3},\cdots,\zeta_{1}-\zeta_{n})\\
 & +(0,\sum_{l>2}(\zeta_{l}-\zeta_{2}),\zeta_{2}-\zeta_{3},\cdots, \zeta_{2}-\zeta_{n})+\cdots\\
 & +(0,0,\cdots,\sum_{l>n-1}(\zeta_{l}-\zeta_{n-1}),\zeta_{n-1}-\zeta_{n})\\
 & =\sum_{i=1}^{n}\sum_{l>i}^{n}(0,\cdots,0,\underbrace{-\zeta_{i}+\zeta_{j}}_{i-\text{th}},0,\cdots0,\underbrace{\zeta_{i}-\zeta_{j}}_{j-\text{th}},0,\cdots,0) \\
 &  =\sum_{i<j}(0,\cdots,0,\underbrace{-\zeta_{i}+\zeta_{j}}_{i-\text{th}},0,\cdots0,\underbrace{\zeta_{i}-\zeta_{j}}_{j-\text{th}},0,\cdots,0).
\end{align*}
\end{proof}

\begin{proof}[\textbf{Proof of Corollary \ref{prop:2p-both}}]\label{proof:2p-both}
(i)  This immediately follows from Proposition \ref{prop:n-p-pop}. (ii) From the second part of Corollary \ref{prop:2p-both},  $ (s^*_1, s^*_2) \in (\arg \max_{s_1} \zeta_2(s_1), \arg \max_{s_2} \zeta_1(s_2))$  is a Nash equilibrium.
If there are two distinctive maximizers, then since the set of maximizers is convex, there exist infinitely many Nash equilibria, contradicting Condition \hyperlink{con-N}{\textbf{(N)}} again by Lemma 2.2 in \citet{Quint&Shubik97}. Thus, the maximizer is unique and constitutes the strictly dominant Nash equilibrium.
\end{proof}

\begin{proof}[\textbf{Proof of Proposition \ref{prop:norm-zero-ci}}]\label{proof:norm-zero-ci}
Let  $ d \sigma_i(s_i) = \frac{1}{m(S_i)} d m_i (s_i)$ be player $i$' uniform mixed strategy. We define a uniform mixed strategy profile as a product measure of uniform mixed strategies: i.e.,
\[
	d \sigma (s) = \prod_i d \sigma_i(s_i).
\]
Let $i$ and $s_{i}$ be fixed. We show that
\[
	f^{(i)}(s_{i},\sigma_{-i})=0.
\]
Then, the desired result follows since $f^{(i)}(s_{i},\sigma_{-i})=0=f^{(i)}(\sigma_{i},\sigma_{-i})$
for all $i$ and $s_{i}$; hence, $f^{(i)}(\sigma_{i},\sigma_{-i})= \max_{s_i} f^{(i)}(s_i, \sigma_{-i})$ for all $i$. First, by the definition of the mixed strategy extension,
\[
f^{(i)}(s_{i},\sigma_{-i})=\int_{s_{-i}\in S_{-i}} f^{(i)}(s_{i},s_{-i})\prod_{l\neq i} d \sigma_l (s_l).
\]
If $f$ is a zero-sum normalized game, then
\[
f^{(i)}(s_{i},\sigma_{-i})= - \int_{s_{-i}\in S_{-i}} \sum_{j \neq i} f^{(j)}(s_{j},s_{-j}) \prod_{l\neq i} d \sigma_l (s_l) = - \sum_{j \neq i} \int_{s_{-i}\in S_{-i}}  f^{(j)}(s_{j},s_{-j}) \prod_{l\neq i} d \sigma_l (s_l) =0
\]
where the last equality follows from the normalization, $\int_{s_l \in S_l} f^{(l)}(s_l, s_{-l}) d \sigma_l (s_l) =0$ for all $l$ and Fubini's Theorem. If $f$ is an identical interest game, then similarly
\[
f^{(i)}(s_{i},\sigma_{-i})=\int_{s_{-i}\in S_{-i}} v(s_{i},s_{-i})\prod_{l\neq i} d \sigma_l (s_l) =0
\]
where the last equality again follows from the normalization, $\int_{s_l \in S_l} v(s_l, s_{-l}) d \sigma_l (s_l) =0$ for all $l$.
Thus, we obtain the desired result.
\end{proof}

\renewcommand{\thesection}{C}
\section{Details for Section \ref{sec:app} \label{appen:app}}

\subsection{Finite strategy games}
\begin{table}[t]
\centering
\scalefont{0.9}

\begin{tabular}{c|c|c|c}
\hline
 & Identity payoff   & Zero-sum  &   Both Potential  \\
 & Normalized  & Normalized  &  and Zero-sum \\
 & $\mathcal{I}\cap\mathcal{N}$  & $\mathcal{Z}\cap\mathcal{N}$  & $\mathcal{B}$ \\
 \hline
Dimensions & $\frac{(l-1)l}{2}$  & $\frac{(l-2)(l-1)}{2}$  & $2 l-1$ \\
\hline
Basis Games
& $\begin{pmatrix}
    1 & -1\\
    -1 & 1
\end{pmatrix}$

& $\begin{pmatrix}
    0 & -1 & 1\\
    1 & 0 & -1\\
    -1 & 1 & 0
\end{pmatrix}$

& $\begin{pmatrix}
    1 & 1\\
    0 & 0
\end{pmatrix}$
 $\begin{pmatrix}
    1 & 0\\
    1 & 0
\end{pmatrix}$
$\begin{pmatrix}
    0 & 1\\
    0 & 1
\end{pmatrix}$
\\
\hline
\end{tabular}
\caption{\textbf{Dimensions of subspaces and basis games for two-player symmetric games}}
\label{tab:basis}
\end{table}

\begin{lem} \label{appen:lem-c1}
We have the following results:\\
  (i) The set of games $\{S^{(ij)}\}_{i=1,\cdots, l, j>i}$ forms a basis set for $\mathcal{I} \cap \mathcal{N}$ \\
  (ii) The set of games $\{Z^{(ij)}\}_{i=2,\cdots, l, j>i}$ forms a basis set for $\mathcal{I} \cap \mathcal{N}$ \\
  (iii) The set of games $\{D^{(i)}\}_{i=1,\cdots, l-1}, \{E^{(i)}\}_{i=1,\cdots, l}$ forms a basis set for $\mathcal{B}$
\end{lem}
\begin{proof}
 (i) We note that there are precisely  $\frac{l(l-1)}{2}$ number of different $S^{(ij)}$'s. Thus, we only need to show that these $S^{(ij)}$'s are independent.  Let $S$ be
 \[
    S:= \sum_{i=1}^{l} \sum_{j=i+1}^l \alpha^{(ij)} S^{(ij)}.
 \]
 Then it is easy to check that $S_{ij}=\alpha^{(ij)}$.Thus if $S=\mathbf{O}$, then $\alpha^{(ij)}=0$ for all $i, j$. \\
 (ii) Again we note that there are precisely $\frac{(l-2)(l-1)}{2}$ number of different $Z^{(ij)}$'s. Let $Z$ be
 \[
    Z:= \sum_{i=2}^{l} \sum_{j=i+1}^l \zeta^{(ij)} Z^{(ij)}.
 \]
Then it is also easy to check that $Z_{ij}=-\zeta^{(ij)}$.Thus if $Z=\mathbf{O}$, then $\zeta^{(ij)}=0$ for all $i, j$. \\
(iii) Again we note that there are precisely $l-1$ number of different $D^{(i)}$'s and $l$ number of different $E^{(i)}$'s. Let $K$ be
\[
    K=  \sum_{i=1}^{l-1}\delta_{i} D^{(i)}+ \sum_{i=1}^{l}\eta_{i} E^{(i)}
\]
Then if $K= \mathbf{O}$, then $\eta_i =0$ for all $i$ (because the last row of $K$ is given by $(\eta_1, \cdots, \eta_l)$) and this, in turn, implies that $\delta_i=0$ for all $i$.
\end{proof}
We would have the following results.
\begin{prop} \label{appen:prop:app1}
  We have the following results: \\
  (i) Suppose that $\gamma_{ij}<0$ for all $i,j$. Then $\#(G) = 1$. \com{In addition, if $D=\mathbf{0}$, then the uniform mixed strategy is a unique NE.} \\
  (ii) Suppose that $\gamma_{ij}>0$ for all $i,j$. Suppose that $\delta_i\geq 0$ for all $i$ and $\underline \gamma > \bar \delta + \bar \zeta$. Then $\#(G) = 2^l-1$.
\end{prop}

\begin{proof}[\textbf{Proof of Proposition \ref{appen:prop:app1} (ii)}]
  We will show that $G$ satisfies the total band wagon property defined by \citet{Kandori98}. Then for any $A \subset \{1, 2, \cdots, l \}$, there exists a unique Nash equilibrium, which is completely mixed in $A$ and thus there exist precisely $2^l-1$ Nash equilibria (See  \citet{Kandori98}). Thus, we will show that for all $q \in \Delta$, $BR(q) \subset \Sigma_q$, where $BR(q)$ is the set of all pure strategy best responses for $q$.  Suppose that there exists $q \in \Delta$ such that $BR(q) \not \subset \Sigma_q$. Then we must have $\Sigma_q \neq \{1, \cdots, l\}$ (the set of all pure strategies) and there exists $k \not \in \Sigma_q$ such that
  \[
    e_k \cdot G q \geq q \cdot G q
  \]
  We define  $\gamma_{ji}=\gamma_{ij}$ for $j > i$. First observe that we have
  \[
    q \cdot S q = \sum_{i <j } \gamma_{ij}(q_i - q_j)^2, \,\,\,\, e_k \cdot S q = \sum_{j<k} \gamma_{jk} (q_k - q_j) + \sum_{j>k} \gamma_{kj}(q_k - q_j) = \sum_{j \neq k} \gamma_{kj}(q_k -q_j)
  \]
  Next we define $\zeta_{ji} = - \zeta_{ij} \text{ for } j > i$. Again observe that we have
  \[
    q \cdot Z q = 0,  \,\,\, \,\, e_1 \cdot S q = \sum_{i<j} \zeta_{ij} (q_j - q_i), \,\,\, \,\,e_k \cdot S q = \sum_{j \neq k } \zeta_{ij} (q_1 - q_j)
  \]
  Thus
  \begin{align*}
     & e_1 \cdot S q \leq \sum_{i<j} \max_{i<j} |\zeta_{ij}| |q_j - q_i| \leq  \max_{i<j} |\zeta_{ij}| \sum_{i<j} \leq  \bar \zeta \\
     & e_k \cdot S q \leq \sum_{j \neq k} \max_{i<j} |\zeta_{ij}| |q_1 - q_j| \leq  \max_{i<j} |\zeta_{ij}| \sum_{i<j} \leq  \bar \zeta
  \end{align*}
  Next we let $d=(\delta_1, \cdots, \delta_l)^T$ and find that
  \[
    q \cdot D q = q \cdot d \geq 0,  \,\,\,\, e_k \cdot D q = \delta _k
  \]
  since $\delta_i \geq 0$ for all $i$.
    Then since $k \not \in \Sigma_q$ so $q_k=0$. Thus
  \begin{align*}
     & 0\leq q \cdot S q + q \cdot Z q  + q \cdot D q  = q \cdot G q \leq e_k \cdot G q = \sum_{j \neq k} \gamma_{kj}(q_k -q_j)+\bar \zeta + \bar \delta \\
     &     =- \sum_{j \neq k}  \gamma_{kj} q_j + \bar \zeta + \bar \delta  \leq -\sum_{j \in \Sigma_q}\gamma_{kj} q_j + \bar \zeta + \bar \delta \leq - \underline \gamma + \bar \zeta + \bar \delta <0
  \end{align*}
  which is a contradiction. The last inequality in the above follows from
  \[
   \sum_{ j \in \sum_q } \gamma_{kj} q_j \geq \sum_{ j \in \sum_q } \min_{ j \in \sum_q} \gamma_{kj} q_j \geq \sum_{ j \in \sum_q } \min_{ j \neq k} \gamma_{kj} q_j  \geq \min_{j \neq k} \gamma_{kj} \sum_{j \in \sum_q} q_j = \min_{j \neq k} \gamma_{kj} \geq \underline \gamma
  \]

\end{proof}

To show (i) of Proposition \ref{appen:prop:app1}, we recall the following definitions (from \citet{Hofbauer09}).
\begin{defn}
  We say that  \\
  (i) a symmetric game $G$ is stable if $(q-p) \cdot G (q-p) \leq 0 \text{ for all } p, q \in \Delta$ \\
  (ii) a symmetric game $G$ is strictly stable if $(q-p) \cdot G (q-p) < 0 \text{ for all } p \neq q \in \Delta$ \\
  (iii) a symmetric game $G$ is null-stable if $(q-p) \cdot G (q-p) = 0 \text{ for all } p, q \in \Delta$
\end{defn}
Next we have the following well-known observation.
\begin{lem}
If $p$ satisfies
\[
    (q-p) \cdot G q <0 \text{ for all } q \neq p \in \Delta
\]
then $p$ is a unique Nash equilibrium for  a symmetric game, $G$.
\end{lem}
\begin{proof}
  Since $G$ is finite, there exist a Nash equilibrium, say $p'$. We will show that $p' =p$. Suppose that $p' \neq p$. Then we find
  \[
    p \cdot G p' > p' \cdot G p'
  \]
  which shows that $p'$ is not a Nash equilibrium, a contradiction. Thus we must have $p'=p$. And this also shows that there cannot exist any other Nash equilibrium.
\end{proof}

We have the following characterization for the strict stability of $G$.
\begin{lem}
  Suppose that $G$ is given by \eqref{eq:decomp-rep}. \\
  (i) $G$ is strictly stable if $\gamma_{ij} < 0$ for all $i < j$. \\
  (iii) $G$ is null stable if  $\gamma_{ij} = 0$ for all $i < j$.
\end{lem}
\begin{proof}
(i) Let $T\Delta$ be the tangent space of $\Delta$ and $ z \neq 0$ and $z \in T \Delta$. Then since $G=S+Z+B$, $Bz=\mathbf{0}$,and $z\cdot Z z = 0$, we have
\[
    z \cdot G z = z \cdot S z = z \cdot \sum_{i<j} \gamma_{ij} S^{(ij)} z = \sum_{i<j} z \cdot S^{(ij)} z = \sum_{i<j} \gamma_{ij}(z_i - z_j)^2 \leq 0
\]
because $\gamma_{ij} <0$ for all $i< j$.
If $\sum_{i<j} \gamma_{ij}(z_i - z_j)^2 = 0$, $\gamma_{ij}(z_i - z_j)^2=0$ for all $i<j$ and thus $z_i - z_j=0$ for all $i>j$ which is a contradiction to $z \neq 0$. Thus we have $z \cdot G z <0$. (ii) Let $z \in T\Delta$. If $\gamma_{ij}=0$, then we again have
\[
    z \cdot G z = \sum_{i<j} \gamma_{ij}(z_i - z_j)^2 =0
\]
\end{proof}

\begin{proof}[\textbf{Proof of Proposition \ref{appen:prop:app1} (i)}]
Suppose that $\gamma_{ij} <0$ for all $i>j$. Since $G$ is a finite game, there exists a NE, $p^*$, for $G$. Since $G$ is strictly stable, for all $q \neq p^*$, we have
\[
    (q-p^*) \cdot G (q-p^*) < 0 \text{ for all } q \neq p^* \in \Delta
\]
Thus
\[
    (q - p^*) \cdot Gq < (q - p^*) \cdot G p^* \leq 0
\]
where the last inequality follows from $p^*$ is a NE. Thus we find that $\#(G)=1$.
\end{proof}

\subsection{Contest games}

First note that $s=(0,0, \cdot, 0)$ cannot be a Nash equilibrium since any player $i$ can deviate to $s_i>0$. Thus we let
\[
    S = \{ (s_1, \cdots, s_n): s_i \geq 0 \text{ for all } i,\text{  and }  s_j > 0 \text{ for some } j \}.
\]
\begin{lem} \label{lem:contest1}
    Let $i$ be fixed and $s_i \geq 0$. Then $w^{(i)}(s_i, \cdot) : S_{-i} \rightarrow \mathbb{R}$ is convex.
\end{lem}
\begin{proof}
We will show that $p^{(i)}(s_i, \cdot) : S_{-i} \rightarrow \mathbb{R}$ is convex and then the desired result follows. Suppose that $s_i > 0$. Define $g: s_{-i} \mapsto \sum_{l \neq i} s_l$ and $h: t \mapsto \frac{s_1}{s_1 +t}$.  Then $g$ is convex, $h$ is convex and decreasing, thus $p^{(i)}(s_i, \cdot)$ is convex.  If $s_i=0$, then
\[
    p^{(i)} (0, s_{-i}) = \begin{cases}
                            \frac{1}{n}, & \mbox{if } s_{-i} = 0 \\
                            0, & \mbox{otherwise}.
                          \end{cases}
\]
Thus $p^{(i)} (0, \cdot)$ is convex for all $s_{-i} \neq 0$. Thus we obtain the desired result.
\end{proof}
Since it is known that the rent-seeking game admits a Nash equilibrium, Proposition \ref{prop:zero-convex} and Lemma \ref{lem:contest1} show that the set of Nash equilibrium for the rent-seeking game is convex.
Let $b_i^\circ (s_{-i})$ be the best response when an interior solution occurs. That is,  $b_i^\circ (s_{-i})$ satisfies
\[
    c_i (b_i^\circ(s_{-i}) + \sum_{l \neq i } s_l )^2 = \sum_{l \neq i } s_l.
\]
Then
\[
    \Phi_f(s) = \sum_{l=i } w^{(i)} (\max \{ b_i^\circ (s_{-i}), 0 \}, s_{-i})
\]
For $P \subset \{1, \cdots, n\}$ such that $|P|\geq 2$, we define
\[
    w_P^{(i)}(s_i, s_{-i}) =(p^{(i)}(s_i, s_{-i}) - \frac{1}{|P|} )  - \frac{1}{|P|-1} \sum_{j \neq i, j \in P} (c_i s_i - c_j s_j)
\]
for $s \in \prod_{i=1} S_i$.

\begin{lem} \label{lem:contest2}
    Suppose that $s^*$ is a Nash equilibrium and $s_i ^* >0$ for all $i \in P$ and $s_i^*=0$ for all $i \not \in P$ where $P \subset \{1, \cdots, n \}$. Let $s_P^* = (s_i^*)_{i \in P}$. Then we have
    \[
        \Phi_f(s^*) = \sum_{i \in P} w^{(i)}_P( b_i^\circ(s^*_{P, -i}), s_{P,-i}^*)
    \]
\end{lem}
\begin{proof}
Let $P \subset \{1, \cdots, n \}$ such that for all $i \in P$,  $s_i >0, b_i(s_{-i})>0$ and for all $i \not \in P$, $s_i = 0, b_i(s_{-i})=0$. Then we have
\begin{align*}
  \sum_{i=1}^n w^{(i)}(b_i(s_{-i}), s_{-i})= &  \sum_{i \in P} w^{(i)}(b^0_i(s_{-i}), s_{-i}) + \sum_{i \not \in P} w^{(i)}(0, s_{-i})\\
  = & \sum_{i \in P} [ p^{(i)}(b_i^\circ (s_{-i}), s_{-i}) -c_i b^\circ_i(s_{-i})) + \frac{1}{n-1} \sum_{j \neq i, j \in P} c_j s_j]  \\
    + &  \sum_{i \not \in P} \frac{1}{n-1} \sum_{j \neq i, j \in P } c_j s_j -1 \\
  = & \sum_{i \in P} w^{(i)}_P( b_i^\circ(s_{P, -i}), s_{P,-i})
\end{align*}
where we use $b^\circ_i(s_{-i}) =b^\circ_i(s_{P, -i})$, since $b^\circ_i(s_{-i})$ depends $\sum_{l \neq i}s_l$.
Using this, we obtain the desired result.
\end{proof}
Lemma \ref{lem:contest2} leads us to define
\begin{align}
  \Phi^\circ_P(s) := & \sum_{i \in P} w^{(i)}_P( b_i^\circ(s_{-i}), s_{-i}) \notag \\
    = & \frac{1}{|P|-1}\sum_{i \in P} c_i  (\sum_{i \in P } \sum_{l \neq  i} s_l) - 2 \sum_{i \in P } \sqrt{c_i} \sqrt{ \sum_{ l \neq i } s_l } + |P|-1 \label{eq:phi-circ-p}
\end{align}
for $s \in S(P)$.

\begin{lem} \label{lem:contest3}
    $\Phi^\circ_P(s): S(P)  \rightarrow \mathbb{R}$ is strictly convex.
\end{lem}
\begin{proof}
  From \eqref{eq:phi-circ-p}, it is enough to consider the following function:
  \[
    \Psi(s):= \sum_{i=1}^{n} \alpha_i h( \sum_{l \neq i} s_l)
  \]
  where $\alpha_i >0$ and $h$ is strictly convex. We will show that $\Psi$ is strictly convex. Let $s, t \in S_+$ and $s \neq t$ and $\rho \in (0, 1)$. Then for some $k$, $\sum_{l \neq k} s_l \neq \sum_{l \neq k} t_l$. Otherwise, if $\sum_{l \neq i} s_l =\sum_{l \neq i} t_l$, then $\sum_{l} s_l = \sum_l t_l$, which again implies $s_i=t_i$, a contradiction. Thus from the strict convexity of $h$, we have
  \[
       h(  (1-\rho) \sum_{l\neq k}s_l + \rho \sum_{l\neq k} t_l) > (1-\rho) h(\sum_{l\neq k}s_l) + \rho h(\sum_{l\neq k} t_l)
  \]
and
  \[
    \Psi( (1-\rho) s + \rho t) = \sum_{i=1}^{n} \alpha_i h( \sum_{l \neq i} (1-\rho) s_l + \rho t_l) > \sum_{i=1}^{n} \alpha_i (1-\rho) h(\sum_{l\neq i}s_i) + \rho h(\sum_{l\neq i} t_l) =(1-\rho)\Psi(s) + \rho \Psi(t)
  \]
\end{proof}

\begin{lem} \label{lem:contest4}
    Suppose that $s^*$ and $t^*$ are Nash equilibria for a rent-seeking game defined in \eqref{eq:r-s-games} such that $s_i^*, t_i^* >0$ for all $i \in P$ and $s_i^*, t_i^* =0$ for all $ i \not \in P$ for some $P \subset \{1,\cdots, n\}$. Then $s^*=t^*$.
\end{lem}
\begin{proof}
  Suppose that $s^*$ and $t^*$ are Nash equilibria for $\Gamma^{(n)}$ such that $s_i^*, t_i^* >0$ for all $i \in P$ and $s_i^*, t_i^* =0$ for all $ i \not \in P$. Let $s_P^*:= (s_i^*)_{i \in P}$ and $t_P^*:= (t_i^*)_{i \in P}$.
    Then we have
    \[
        0 = \Phi(s^*) = \Phi^\circ_P(s_P^*) \text{ and } 0 = \Phi(t^*) = \Phi^\circ_P(t_P^*).
    \]
    Since $\Phi^\circ_P(s_P) \geq 0$ for all $s_P \in \prod_{i \in P} S_i$ (This is to be shown) and the strict convexity of $\Phi_P^\circ(s)$ implies that the minimum is unique. We have $s_P^* = t_P^*$, and thus $s^* = t^*$.
\end{proof}

\begin{prop} \label{appen:prop-app2}
    The Nash equilibrium for the rent-seeking game defined in \eqref{eq:r-s-games} is unique.
\end{prop}
\begin{proof}
  Suppose that $s^*$ and $t^*$ such that $s^* \neq t^*$ are Nash equilibria.  Let $P':=\{i: s_i^*>0 \}$ and $P'':= \{i: t_i^*>0 \}$. Then from Lemma \ref{lem:contest4}, we must have $P' \neq P$. Since the set of Nash equilibria is convex by Lemma \ref{lem:contest1}, $\rho s^* + (1-\rho) t^*$ is a Nash equilibria for all $\rho \in [0, 1]$. Then for $0 < \rho < 1$,  $(\rho s^* + (1-\rho) t^*)_i >0 $ if $ i \in P'$ and $(\rho s^* + (1-\rho) t^*)_j >0 $ if $ j \in P''$. Thus there are infinitely many Nash equilibrium for the set $P' \cup P''$, which is contradiction to Lemma \ref{lem:contest4}.
\end{proof}

\renewcommand{\thesection}{D}
\section{Existing decomposition results}\label{sec:existing}

Our decomposition methods extend two kinds of existing results:  (i) \citet{Kalai10}, (ii) \citet{HandR11, Candogan2011}. First, \citet{Kalai10} decompose normal form games with incomplete information and study the implications for Bayesian mechanism designs. Their decomposition is based on the orthogonal decomposition $\mathcal{L}=\mathcal{I} \oplus\mathcal{Z}$ in equation (\ref{eq:1st-two-d1}).
Second, \citet{HandR11} similarly provide decomposition results based on the orthogonality between identical interest and zero-sum games and between normalized and non-strategic games, mainly focusing on finite games. \citet{Candogan2011} decompose finite strategy games into three components: a potential component, a nonstrategic component, and a harmonic component. When the numbers of strategies are the same for all players, harmonic components are the same as zero-sum normalized games, and their harmonic games, in this case, refer to games that are strategically equivalent to zero-sum normalized games. Also, their potential component is obtained
by removing the non-strategic component from the potential part ($\mathcal{I}+\mathcal{E}$)
of the games. Note that we can change our definition of zero-sum normalized games to their definition of harmonic games, with all the decomposition results remaining unchanged. Thus, their three-component decomposition of finite strategy
games follows from Theorem \ref{thm:main}, $\mathcal{L}=(\mathcal{I}+\mathcal{E})\oplus(\mathcal{Z}\cap\mathcal{N})$
(see the proof of Corollary \ref{cor: Can} for more detail).
\begin{cor}
\label{cor: Can}We have the following decomposition.
\[
\mathcal{L}=\underbrace{((\mathcal{\mbox{\ensuremath{\mathcal{I}}}}+\mathcal{E})\cap\mathcal{N})}_{\text{Potential Component}}\oplus\underbrace{\mathcal{E}}_{\substack{\text{Nonstrategic}\\
\text{Component}
}
}\oplus\underbrace{(\mathcal{Z}\cap\mathcal{N})}_{\substack{\text{Harmonic}\\
\text{Component}
}
}
\]

\end{cor}
\begin{proof}
This proof follows from Theorem \ref{thm:main} by showing that $ $$((\mathcal{I}+\mathcal{E})\cap\mathcal{N})\oplus\mathcal{E}=\mathcal{I}+\mathcal{E}$.
First, observe that $(\mathcal{I}+\mathcal{E})\cap\mathcal{N}\subset\mathcal{I}+\mathcal{E}$,
which implies that $((\mathcal{I}+\mathcal{E})\cap\mathcal{N})\oplus\mathcal{E}\subset\mathcal{I}+\mathcal{E}$.
Now, let $f\in\mathcal{I}+\mathcal{E}$. Then, $f=g+h,$ where $g\in\mathcal{I}$,
$h\in\mathcal{E}$, and $g=(v,v,\cdots,v).$ Then, by applying the
map, $\mathbf{P}$, we find that $f=\mathbf{P}(f)+(I-\mathbf{P})(f)$. Obviously,
$\mathbf{P}(f)\in\mathcal{E}$. In addition, $(I-\mathbf{P})(f)=(I-\mathbf{P})(g)=(v-T_{1}v,\,v-T_{2}v,\cdots,\,v-T_{n}v)\in\mathcal{I}+\mathcal{E}.$
Thus, $(I-\mathbf{P})(f)\in(\mathcal{I}+\mathcal{E})\cap\mathcal{N}$.
\end{proof}
\com{\noindent Note that Corollary \ref{cor: Can} not only reproduces the result of \citet{Candogan2011}, when the number of strategies of the players is the same, but also extends it to the space of games with continuous strategy sets.
}

\citet{Ui00} provides the following characterization for potential
games:
\begin{equation} \label{eq:Ui-con}
f\text{ is a potential game if and only if }f^{(i)}=\sum_{\substack{M\subset N\\
M\ni i
}
}\xi_{M}\textnormal{ for some\,}\{\xi_{M}\}_{M\subset N}\textnormal{ for all }i
\end{equation}
where $\xi_{M}$ depends only on $s_{l}$, with $l\in M$. Let
\[
    \mathcal{D} := \{ f \in \mathcal{L} : f^{(i)}(s) := \sum_{l \neq i} \zeta_l (s_{-l}) \text{ for all } i \}.
\]
\noindent From our decomposition results, we have $\mathcal{D} \subset \mathcal{I}+\mathcal{E}$ and $\mathcal{E} \subset \mathcal{I} + \mathcal{D}$. In particular, the second inclusion holds because
\[
    \zeta_{i}(s_{-i}) =\sum_{l=1}^{n} \zeta_l (s_{-l}) - \sum_{l \neq i} \zeta_{l}(s_{-l}).
\]
Thus, $\mathcal{D} \subset \mathcal{I}+\mathcal{E}$ implies  that $\mathcal{I}+\mathcal{D}+\mathcal{E} \subset \mathcal{I}+\mathcal{E}$ and $\mathcal{E} \subset \mathcal{I} + \mathcal{D}$ implies $\mathcal{I}+\mathcal{D}+\mathcal{E} \subset \mathcal{I}+\mathcal{D}$. From this, we find
\begin{equation}\label{eq:ui-con-show}
  \mathcal{I}+\mathcal{D}= \mathcal{I}+\mathcal{D}+\mathcal{E}=  \mathcal{I}+\mathcal{E}
\end{equation}
Note that all games in $\mathcal{I}$ and in $\mathcal{D}$ satisfy Ui's condition in \eqref{eq:Ui-con}; hence, games in $\mathcal{I} + \mathcal{D}$ satisfy Ui's condition. Then, equalities in \eqref{eq:ui-con-show} show that the condition in \eqref{eq:Ui-con} is a necessary condition for potential games. The sufficiency of Ui's
condition is deduced by adding the non-strategic game
\[
(\sum_{\substack{M\subset N\\
M\not\ni1
}
}\xi_{M},\sum_{\substack{M\subset N\\
M\not\ni2
}
}\xi_{M},\cdots,\sum_{\substack{M\subset N\\
M\not\ni n
}
}\xi_{M})
\]
to game $f$ satisfying Ui's condition.

As explained in the main text, \citet{Sandholm10} decomposes $n$-player finite strategy games into
$2^{n}$ components using an orthogonal projection. When the set of games consists of symmetric games with $l$ strategies, the orthogonal projection is given by $\Gamma:=I-\frac{1}{l}\mathbf{1}\mathbf{1}^{T}$
, where $I$ is the $l \times l$ identity matrix and $\mbox{\textbf{1}}$ is the column vector consisting of all
1's. Using $\Gamma$, we can,
for example, write a given symmetric game, $A$, as
\begin{equation}
A=\underset{=(\mathcal{I}\cap\mathcal{N})\oplus(\mathcal{Z}\cap\mathcal{N})}{\underbrace{\Gamma A\Gamma}}+\text{ }\underbrace{(I-\Gamma)A\Gamma+\Gamma A(I-\Gamma)+(I-\Gamma)A(I-\Gamma)}_{=\mathcal{B}}.\footnote{In fact, for two player symmetric game, using Table \ref{tab:component} we can verify that
\[
f_{\mathcal B}=(I  - (I-T_1)(I-T_2))f^{(1)}, \quad  f_{\mathcal{I}\cap\mathcal{N}}+ f_{\mathcal{Z} \cap \mathcal{N}}=(I-T_1)(I-T_2) f^{(1)}.
\]}\label{eq:sand_decomp}
\end{equation}
Thus, our decompositions show that $\Gamma A\Gamma$ can be decomposed
further into games with different properties---identical interest normalized
games and zero-sum normalized games---and every game belonging to the
second component in (\ref{eq:sand_decomp}) is strategically equivalent to both an identical interest game and a zero-sum game. \citet{Sandholm10}
also shows that a two-player game, ($A,B)$, is potential if and only
if $\Gamma A\Gamma=\Gamma B\Gamma.$ If $P=(P^{(1)},P^{(2)})$ is a non-strategic
game, it is easy to see that $\Gamma P^{(1)}=O$ and $P^{(2)}\Gamma=O$,
where $O$ is a zero matrix. Thus,
the necessity of the condition $\Gamma A\Gamma=\Gamma B\Gamma$ for potential
games is obtained. Conversely, if $\Gamma A\Gamma=\Gamma B\Gamma$,
then game $(A,B)$ does not have a component belonging to $\mathcal{Z}\cap\mathcal{N}$
because $(\Gamma A\Gamma,\Gamma B\Gamma)\in(\mathcal{I}\cap\mathcal{N})\oplus(\mathcal{Z}\cap\mathcal{N}).$
Thus, ($A$,$B$) is a potential game.

\end{document}